\documentclass[12pt,twoside]{article}

\usepackage{float}
\usepackage{graphicx}
\usepackage{epstopdf}
\usepackage{graphicx}
\usepackage{epic}
\usepackage{multirow}
\usepackage{pst-poly}  
\usepackage{pst-plot}  
\usepackage{pst-poly}  
\usepackage{xcolor}
\usepackage[all,pdf]{xy}

\renewcommand{\paragraph}{\roman{paragraph}}
\usepackage[a4paper]{geometry}
\makeatletter
\renewcommand\title[1]{\gdef\@title{\reset@font\Large\bfseries #1}}
\renewcommand\section{\@startsection {section}{1}{\z@}%
	{-3.5ex \@plus -1ex \@minus -.2ex}%
	{2.3ex \@plus.2ex}%
	{\normalfont\large\bfseries}}
\renewcommand\subsection{\@startsection{subsection}{2}{\z@}%
	{-3ex\@plus -1ex \@minus -.2ex}%
	{1.5ex \@plus .2ex}%
	{\normalfont\normalsize\bfseries}}
\renewcommand\subsubsection{\@startsection{subsubsection}{3}{\z@}%
	{-2.5ex\@plus -1ex \@minus -.2ex}%
	{1.5ex \@plus .2ex}%
	{\normalfont\normalsize\bfseries}}

\def\@runningauthor{}\newcommand{\runningauthor}[1]{\def\runningauthor{#1}}
\def\@runningtitle{}\newcommand{\runningtitle}[1]{\def\runningtitle{#1}}

\renewcommand{\ps@plain}{%
	\renewcommand{\@evenhead}{\footnotesize\scshape \hfill\runningauthor\hfill}
	\renewcommand{\@oddhead}{\footnotesize\scshape \hfill\runningtitle\hfill}}
\newcommand{\Z}{\mathbb{Z}}
\newcommand{\F}{\mathbb{F}}

\newcommand {\C}{{\mathcal{C}}}
\newcommand {\ccc}{{\mathbf{c}}}

\g@addto@macro\bfseries{\boldmath}

\makeatother

\usepackage{amsthm,amsmath,amssymb}
\usepackage{cite}
\usepackage{graphicx}

\usepackage[colorlinks=true,citecolor=black,linkcolor=black,urlcolor=blue]{hyperref}

\theoremstyle{plain}
\newtheorem{theorem}{Theorem}
\newtheorem{lemma}[theorem]{Lemma}

\newtheorem{proposition}[theorem]{Proposition}

\theoremstyle{definition}

\newtheorem{example}[theorem]{Example}

\theoremstyle{remark}

\title{Constructions of block MDS LDPC codes from punctured circulant matrices
}

\runningtitle{Constructions of block MDS LDPC codes from punctured circulant matrices}

\author{Hongwei Zhu \thanks{Hongwei Zhu is with the School of Mathematics and Information
Science, Guangzhou University, Guangzhou, Guangdong 510006, China, and also with the Tsinghua Shenzhen International Graduate School, Tsinghua University, Shenzhen 518055, China. E-mail: zhwgood66@163.com, hongweizhu@sz.tsinghua.edu.cn.}
	\and Xuantai Wu \thanks{Xuantai Wu is with the School of Mathematical Sciences, Naikai University, Tianjin 300071, China. E-mail: jason.wxt@live.com}	
	\and Jingjie Lv \thanks{Jingjie Lv is with the School of Electrical Engineering \& Intelligentization, Dongguan University of Technology, Dongguan 523808, China. E-mail: juxianljj@163.com}
	\and Qinshan Zhang\thanks{Qinshan Zhang is with the Tsinghua Shenzhen International Graduate School, Tsinghua University, Shenzhen 518055, China. E-mail: zhangqs24@mails.tsinghua.edu.cn}
	\and Shu-Tao Xia\thanks{Shu-Tao Xia is with the Tsinghua Shenzhen International Graduate School, Tsinghua University, Shenzhen 518055, China. E-mail: xiast@sz.tsinghua.edu.cn}
}

\runningauthor{}

\date{}

\begin{document}

	\maketitle
	
	\thispagestyle{empty}
	
	\begin{abstract}
		Low density parity check (LDPC) codes, initially discovered by Gallager, exhibit excellent performance in iterative decoding, approaching the Shannon limit. MDS array codes, with favorable algebraic structures, are codes suitable for decoding large burst errors. The Blaum-Roth (BR) code, an MDS array code similar to the Reed-Solomon (RS) code but has a parity-check matrix prone to $4$-cycles. Fossorier proposed constructing quasi-cyclic LDPC codes from circulant permutation matrices but are not MDS array codes. This paper aims to construct codes that possess both the block MDS property and have no $4$-cycles in the Tanner graph of their parity-check matrices, namely the so-called block MDS LDPC codes. Non-binary block MDS QC codes were first constructed by [Tauz {\it et al. }IEEE ITW, 2025] using circulant shift matrices. We first generate a family of block MDS codes over $\F_2$ from punctured circulant permutation matrices.  Second, we construct a family of block MDS LDPC codes from circulant matrices with column weight $> 1$ (CM$(t)$). Additionally, we present the Moore determinant formula for CM$(t)$s and a sufficient condition to avoid $4$-cycles in CM\((t)\)-QC LDPC codes' Tanner graphs for $t> 1$. We also point out the non-existence of binary block MDS CPM-QC LDPC codes. Compared to the codes constructed in [Li {\it et al. }IEEE TIT, 2023] and [Xiao {\it et al. }IEEE TCOM, 2021], our block MDS LDPC codes show enhanced random-error-correction at a similar code length and rate. Meanwhile, these codes can effectively combat burst errors when considered as array codes. Both of our two types of constructions for block MDS LDPC codes are applicable to the scenario of the binary field.
	\end{abstract}
	
	\section{Introduction}
	Low density parity check (LDPC) codes, originally discovered by Gallager \cite{Gall1}, showcase remarkable performance when iterative decoding is employed. Their performance approaches the Shannon limit extremely closely \cite{MacK,Richardson}. In most cases, LDPC codes are obtained through computer-based searches. This is due to the fact that there are only a limited number of algebraic construction methods available for them. A key aspect of LDPC code construction lies in the utilization of structural construction methods rather than random construction methods. Tanner's graphical representation of LDPC codes, as presented in \cite{TAN0}, has exerted a profound influence on a substantial portion of the contemporary literature in this research domain. Structural construction encompasses graph-theoretic construction approaches, such as the significant works of Margulis \cite{MAR1} and Tanner \cite{TAN1}. Subsequently, Rosenthal and Vontobel \cite{RV1} constructed LDPC codes from Ramanujan graphs. Kou, Lin, and Fossorier \cite{KLF1} constructed LDPC codes based on finite geometry. These structural construction methods can guarantee distance properties, enhance cycle characteristics, and simplify implementation. Array codes \cite{Fan1,Fan2} can be considered as LDPC codes with excellent algebraic structures, and they are also applied in scenarios of decoding codes with large burst errors.
	
	The Blaum-Roth (BR) code \cite{BR1}, being one of the most renowned MDS array codes, has a structure quite similar to that of the Reed-Solomon (RS) code \cite{RS1}. The distinction is that its symbols are in the Galois ring rather than the Galois field. Owing to the complexity of operations in Galois fields, the symbols are typically very small ($8$ to $16$ bits). However, the symbols of array codes can be much larger (on the order of hundreds of bits) because the decoding complexity remains reasonable. An easily conceivable attempt is to directly utilize the BR code as an LDPC code.
	A $q$-ary BR code with sub-packetization level $p-1$ can be regarded as an RS code over the residue ring $\mathbb{F}_q[x]/\langle 1 + x+\cdots + x^{p - 1}\rangle$, where $p$ is a odd prime whcih is not the characteristic of $\F_q$. Its parity-check matrix can be represented as
	{\small\[
		\mathcal{H}=\begin{pmatrix}
			1&1&1&\cdots&1\\
			1&\alpha&\alpha^{2}&\cdots&\alpha^{p - 1}\\
			\vdots&\vdots&\vdots&\ddots&\vdots\\
			1&\alpha^{r - 1}&\alpha^{2(r - 1)}&\cdots&\alpha^{(p - 1)(r - 1)}
		\end{pmatrix}_{r\times n},
		\]}
	where $\alpha$ denote the monomial $x$ in the residue ring $\mathbb{F}_q[x]/\langle 1 + x+\cdots + x^{p - 1}\rangle$. In the vector space $\mathbb{F}_{q}^{(p - 1)\times(p - 1)}$, $\alpha$ corresponds to the following $(p - 1)\times(p - 1)$ matrix:
	{\small\[
		\alpha=\begin{pmatrix}
			0&0&\cdots&0&1\\
			1&0&\cdots&0&1\\
			\vdots&\vdots&\vdots&\vdots&\vdots\\
			0&0&\cdots&1&1
		\end{pmatrix}_{(p-1)\times (p-1)}.
		\]}
	Then $4$-cycles will emerge in the parity-check matrix of the BR code, so the BR code is not suitable to be used as an LDPC code.
	
	Fossorier \cite{FOS1} proposed the well-known construction of binary quasi-cyclic LDPC codes from circulant permutation matrices (CPM-QC LDPC codes) and presented some sufficient and necessary conditions for the Tanner graph of these QC LDPC codes to have a given girth \(g\) \cite[Theorem 2.1]{FOS1}, as well as sufficient and necessary conditions for the girth to be no less than $6$. The parity-check matrix considered by Fossorier has the form:
	{\small\begin{equation}\label{H1}
			H=\begin{pmatrix}
				I(0) & I(0) & \cdots & I(0) \\
				I(0) & I(p_{1,1}) & \cdots & I(p_{1,L - 1}) \\
				\vdots & \vdots & \ddots & \vdots \\
				I(0) & I(p_{J - 1,1}) & \cdots & I(p_{J - 1,L - 1})
			\end{pmatrix},
	\end{equation}}
	where for \(1\leq j\leq J - 1\), \(1\leq l\leq L - 1\), \(I(p_{j,l})\) represents the \(p\times p\) CPM with a one at column \(r + p_{j,l}\text{ mod }p\) for row \(r\), \(0\leq r\leq p - 1\), and zero elsewhere. Fossorier pointed out that the rank of \(H\) in (\ref{H1}) is at most \(Jp - J+1\).
	Let the code with \(H\) as in (\ref{H1}) serving as the parity-check matrix be \(\mathcal{C}_1\). If we consider \(\mathcal{C}_1\) as an array code with a sub-packetization level of \(p\), then \(\mathcal{C}_1\) definitely cannot be an MDS array code. Consequently, the burst-error-correction ability of \(\mathcal{C}_1\) is not in an optimal state. Although \(\mathcal{C}_1\) is not an MDS array code, scholars still attempt to construct LDPC codes that can correct both random errors and burst errors \cite{Li1,Xiao1}. By imposing the modular Golomb ruler condition, Xiao {\it et al.} \cite{Xiao1} studied the burst-error-correction capability of CPM-QC-LDPC codes when the column weight is 2 or more.
	
	Does an MDS array code exist whose parity-check matrix has no $4$-cycles? If such a code exists, then this code may be capable of correcting both random errors and burst errors. When viewed from the perspective of a binary error-correcting code, an array code can possess an LDPC matrix, which can be utilized in the message-passing algorithm. That is to say, it can be applied in soft iterative decoding schemes. For example, it could potentially be used as the outer code in a concatenated coding scheme, correcting the output of a turbo of an LDPC decoder, serving as an alternative to the hard-decision decoding of RS codes. We term such a code a block MDS LDPC code. Tauz {\it et al.} \cite{Tau1} demonstrated a method of using sampling for relaxation in the information reconciliation (IR) step of quantum key distribution (QKD), thus improving the success rate of the IR step. This enables us to create LDPC codes in the form of block MDS QC-LDPC codes that take privacy amplification (PA) into account to utilize this relaxation. Specifically, they proposed first using the generalized Schur's formula \cite{Kovacs1,Silvester} to select circulant shift matrices quasi-cyclic (CSM-QC) codes with the MDS array property over non-binary fields, and then imposing the necessary and sufficient conditions given by Fossorier \cite[Theorem 2.1]{FOS1} on these codes to make the codes free of $4$-cycles.
	\subsection*{Our Contributions}
	In this paper, we construct block MDS LDPC codes from punctured circulant matrices. The specific contributions are as follows:
	\begin{itemize}
		\item [{\rm (1)}]We constructed a family of block MDS LDPC codes from punctured circulant permutation matrices (see Theorem \ref{thmfirst}). This construction is applicable to any finite field, including the binary field.  We pointed out that there is no binary block MDS CPM-QC code without $4$-cycles. The construction of non-binary block MDS LDPC codes proposed by Tauz {\it et al.} \cite[Theorem 1]{Tau1} has been refined into an algebraic construction with a Vandermonde form (See Theorem \ref{Taurefine}). This construction is easy to implement and does not rely on the cumbersome condition of judging the generalized Schur's formula (Lemma \ref{BlockVan}).
		
		\item [{\rm (2)}]We constructed a family of block MDS LDPC codes from circulant matrices with column weight greater than $1$ (see Theorem \ref{thmsec}). This construction is applicable to finite fields of characteristic $2$. From the results of the simulation experiments, under the condition of the similar code length and code rate, the block MDS LDPC codes we constructed have better performance in correcting random errors compared with the works in \cite{Li1, Xiao1} (see Example \ref{ex17} and Example \ref{ex18}). Specifically, the formula for calculating the Moore determinant of CM$(t)$s is presented here (see Lemma \ref{CMsMoo}). A sufficient condition for avoiding $4$-cycles in the Tanner graph of CM($t$)-QC LDPC codes is given (see Lemma \ref{suff1}).
	\end{itemize}
	\subsection*{Organization}
	The structure of the remaining paper is as follows: Section II lays out essential preliminaries and notation. Section III demonstrates how to construct binary block MDS LDPC codes from punctured CPM. Section IV shows how to construct binary block MDS LDPC codes from punctured circulant matrices with column weight greater than $1$.  Finally, Section V offers concluding remarks.
	\section{Preliminaries}
	Throughout this paper, we assume  and fix the following:
	\begin{itemize}
		\item Let $\F_q$ denote the finite field with $q$ elements, and let $\F_q^*=\F_q\backslash\{0\}$.
		\item Let $\Z_n$ be the ring of integers modulo $n$, i.e., $\Z_n=\{0,1,\ldots,n-1\}.$ Let $[n]=\{1,2,\ldots,n\}.$
		\item Let ${A \choose i}$ denote the set of all $i$-subsets of the set $A$.
		\item Let $(a)_N$ denote the integer $a {\rm {~mod~}}N$.
		\item Let $\Sigma^m$ be the vector space of all $m$-tuples over the alphabet $\Sigma$.
		\item Assume that $N$ is a positive integer meeting $\gcd(N,q)=1$.
	\end{itemize}
	
	\subsection{Array codes}
	An $(n, M;N)$ array code over $\Sigma$ is a subset of $|\Sigma|^{N\times n}$ with size $M$, where $N$ is called the sub-packetization level. We refer to $\C$ as a $q$-ary array code $\C$ with parameters $(n,M;N)$ if $\Sigma=\F_q$.

	For any array $\Gamma=[c_{i,j}]_{\tiny{\substack{1\leq i\leq N\\
				1\leq j\leq n}}}$, the array weight of $\Gamma$, denoted by $w_A(\Gamma)$, is defined as follows:
	$$w_A(\Gamma)=\#\{1\leq j\leq n| (c_{1,j},c_{2,j},\ldots, c_{N,j})\neq \mathbf{0}\}.$$

	For an array code $\C$ over $\Sigma$, the minimum array distance of $\C$ is defined as
	$$d_A(\C)=\min\left\{d_A(\ccc_1,\ccc_2)\left|\{\ccc_1,\ccc_2\} \in {\C\choose 2}\right.\right\}.$$
	The following result which from \cite{BR1} presents the Singleton bound for array codes.
	\begin{theorem}\cite{BR1}\label{Singletonbound}
		If $\C$ is an array code $\C$ with parameters $(n,M;N)$ over $\F_q$, then
		\begin{equation}\label{Singleton2}
			\log_{q}M\leq N(n+1-d_A(\C)).
		\end{equation}
	\end{theorem}
	An array code $\C$ satisfying the Singleton bound given by (\ref{Singleton2}) is defined as an MDS array (or block MDS) code with sub-packetization level $N$ over $\F_q$.
	
	Let $(C_j)_{1 \leq j \leq n}$ denote a codeword of an array code $\C$ with sub-packetization level $N$, where $C_j$ corresponds to column vectors in $\F_q^N$ for $1 \leq j \leq n$. These $C_j$ are designated as the nodes of $\C$. The array code $\C$ with sub-packetization level $N$ is expressed in the following parity-check formulation:
	\begin{equation}\label{arrayparity-checkform}
		\C=\left\{\left(C_j\right)_{1\leq j\leq n}\left| \sum_{j=1}^nA_{i,j}C_j=\mathbf{0} {\rm ~for~}1\leq i\leq r\right.\right\},
	\end{equation}
	where $\mathbf{0}\in \F_q^N$ and each $A_{i,j}$ is an $N\times N$ square matrix. For convenience, we set
	\begin{equation}\label{H2}
		H=\left(
		\begin{array}{cccc}
			A_{1,1} & A_{1,2} & \cdots & A_{1,n} \\
			A_{2,1} & A_{2,2} & \cdots & A_{2,n} \\
			\vdots & \vdots & \ddots & \vdots \\
			A_{r,1} & A_{r,2} & \cdots & A_{r,n} \\
		\end{array}
		\right)=(H_1,H_2,\ldots,H_{n}),
	\end{equation}
	where each $H_i$ is an $N(n-k)\times N$ matrix. The matrix $H$ is referred to as the parity-check matrix of the array code $\C$ with sub-packetization level $N$.
	
	For $N=1$, an $(n-k)\times n$ matrix $H$ over $\F_q$ is an MDS matrix if and only if every square submatrix of $H$ is nonsingular. The following proposition presents a property of the parity-check matrices of MDS array codes with sub-packetization level $N$.
	\begin{proposition}\cite{BR1}\label{proiff}
		Let $\C^{\prime}$ be an array code over $\F_q$ with parity-check matrix {\rm(\ref{H2})}. $\C^{\prime}$ is an MDS array code over $\F_q$ if and only if the $N(n-k)$ columns of any $n-k$ matrices $H_i$ form a linearly independent set over $\F_q$.
	\end{proposition}
	\subsection{Cyclic codes and quasi-cyclic codes}
	A linear $[n,k,d]$ code $\mathcal{C}$ over $\mathbb{F}_q$ is defined as a $k$-dimensional subspace of $\mathbb{F}_q^n$ with a minimum distance of $d$.
	A highly significant class of linear codes is the class of cyclic codes. These codes are invariant under cyclic shifts, which are represented by the following transformation:
	\begin{eqnarray*}
		\mathcal{T}:\mathbb{F}_q^n &\longrightarrow& \mathbb{F}_q^n \\
		\mathbf{c}=(c_0,c_1,\ldots,c_{n - 1}) &\longmapsto& \mathcal{T}(\mathbf{c})=(c_{n - 1},c_0,\ldots,c_{n - 2})
	\end{eqnarray*}
	where the transformation $\mathcal{T}$ acts on the coordinates of the vectors in $\mathbb{F}_q^n$. By establishing a correspondence between any vector $(c_0,c_1,\ldots,c_{n - 1})\in\mathbb{F}_q^n$ and the polynomial $c_0 + c_1x+\cdots+c_{n - 1}x^{n - 1}\in\frac{\mathbb{F}_q[x]}{(x^n - 1)}$, any code $\mathcal{C}$ of length $n$ over $\mathbb{F}_q$ corresponds to a subset of the principal quotient ring $\frac{\mathbb{F}_q[x]}{(x^n - 1)}$.
	
	Let $\mathcal{C}=\langle g(x)\rangle$ be a cyclic code, where $g(x)$ is a monic polynomial of the smallest degree among all the generators of $\mathcal{C}$. Then $g(x)$ is unique and is referred to as the generator polynomial of $\mathcal{C}$, and $h(x)=\frac{x^n - 1}{g(x)}$ is denoted as the parity-check polynomial of $\mathcal{C}$.
	
	Quasi-cyclic (QC) codes with index $n_1$ represent a generalization of cyclic codes. For any codeword $\mathbf{c}=(c_0,c_1\ldots,c_{n_1N - 1})\in\mathcal{C}$, if there exists a smallest positive integer $n_1$ such that $\mathcal{T}^{n_1}(\mathbf{c})\in\mathcal{C}$, then the linear code $\mathcal{C}$ is called a quasi-cyclic code with index $n_1$. The QC code $\mathcal{C}$ can be regarded as an $\frac{\mathbb{F}_q[x]}{(x^N - 1)}$-submodule of $\left(\frac{\mathbb{F}_q[x]}{(x^{N}-1)}\right)^{n_1}$.
	
	We define an $\mathbb{F}_q$-module isomorphism $\phi$ from $\mathbb{F}_{q}^{n_1N}$ to $\left(\frac{\mathbb{F}_q[x]}{(x^N - 1)}\right)^{n_1}$ as follows:
	\begin{equation*}
		\phi(c_0,\ldots,c_{N - 1},c_{N},\ldots,c_{2N - 1},\ldots,c_{(n_1 - 1)N},\ldots, c_{n_1N - 1})=(c_0(x),c_1(x),\ldots,c_{n_1 - 1}(x)),
	\end{equation*}
	where $c_i(x)=\sum_{t = iN}^{(i + 1)N - 1}c_t x^{t}\in\frac{\mathbb{F}_q[x]}{(x^N - 1)}$ for $i\in\mathbb{Z}_{n_1}$. If $\mathcal{C}$ is generated by $s$ elements $\mathbf{a}_0,\mathbf{a}_1,\ldots,\mathbf{a}_{s - 1}\in\left(\frac{\mathbb{F}_q[x]}{(x^N - 1)}\right)^{n_1}$, where $\mathbf{a}_0=(a_{0,1}(x),\ldots,a_{0,n_1}(x)), \ldots,\mathbf{a}_{s - 1}=(a_{s - 1,1}(x),\ldots,\\ a_{s - 1,n_1}(x)),$ then $\mathcal{C}$ is an $s$-generator QC code with index $n_1$.
	
	
	\subsection{Block matrices, circluant matrices and punctured circulant matrices}
	For an array code $\C$ with the parity-check matrix $H$ of form (\ref{H2}), $\C$ is an MDS array code with parameters $(n,q^{Nk},d_A(\C)=n-k+1;N)$ over $\F_q$ if and only if every $r\times r$ block submatrix
	\begin{equation*}\label{commute1}
		M_{(t_1,t_2,\ldots,t_r)}=\left(
		\begin{array}{cccc}
			A_{1,t_1} & A_{1,t_2} & \cdots & A_{1,t_r} \\
			A_{2,t_1} & A_{2,t_2} & \cdots & A_{2,t_r} \\
			\vdots & \vdots & \ddots & \vdots \\
			A_{r,t_1} & A_{r,t_2} & \cdots & A_{r,t_r} \\
		\end{array}
		\right)
	\end{equation*}
	is invertible, where $\{t_1,t_2,\ldots,t_r\}\in {[n]\choose r}$.
	If the blocks $A_{i,t_j}$ of the block matrix $M_{(t_1,t_2,\ldots,t_r)}$ are pairwise commute, then we have the following method to determine whether the block matrix is invertible.
	
	\begin{lemma}\label{BlockVan}(Generalized Schur's formula)\cite{Kovacs1,Silvester}
		Assume that all blocks of $M_{(t_1,t_2,\ldots,t_r)}$ are pairwise commutative. Then the determinant of $M_{(t_1,t_2,\ldots,t_r)}$ is
		\begin{equation*}
			\left|\sum_{\pi\in {\rm Sym}_r}{\rm sgn} (\pi)A_{1\pi(t_1)}A_{2\pi(t_2)}\cdots A_{r\pi(t_r)}\right|,
		\end{equation*}
		where ${\rm Sym}_r$ is the set of all permutations of the set of $r$ coordinates and ${\rm sgn} (\pi)$ equals $+1$ if $\pi$ is an even permutation and equals $-1$ if $\pi$ is an odd permutation.
	\end{lemma}
	
	A square matrix \(A\) over $\F_q$ is referred to as a circulant matrix with column weight \(t\) when it assumes the following form:
	\begin{equation}\label{CM1}
		A=\begin{pmatrix}
			a_0 & a_1 & \cdots & a_{N - 1} \\
			a_{N - 1} & a_{0} & \cdots & a_{N - 2} \\
			\vdots & \vdots  & \ddots & \vdots \\
			a_{1} & a_{2} & \cdots & a_{0} \\
		\end{pmatrix},
	\end{equation}
	where precisely \(t\) of the coefficients \(a_0,a_1,\ldots,a_{N - 1}\) are non-zero elements in $\F_q$.
	The matrix \(A\) can be represented as \(A = a_0I+a_{N-1}P_N+\cdots+a_{1}P_N^{N - 1}\), where \(I\) denotes the \(N\)-order identity matrix and \(P_N\) is the \(N\times N\) circulant permutation matrix having a one at column \((s + 1)\bmod N\) for row \(s\), with \(s\in[N]\).
	The matrix \(A\) can be put into a one-to-one correspondence with the polynomial \(a(x)=a_0 + a_{N-1}x+\cdots+a_{1}x^{N - 1}\in\frac{\mathbb{F}_q[x]}{(x^N - 1)}\). We term the polynomial \(a(x)\) as the associated polynomial of the circulant matrix \(A\). For convenience, we abbreviate matrix of the form (\ref{CM1}) as CM$(t)$; the matrix of the form $P_N^i$ is abbreviated as CPM; and  the matrix of the form $s_i\cdot P_N^i$ is abbreviated as circulant shift matrix (CSM), where $s_i\in\mathbb{F}_q^*$.
	
	Let \({\rm Pu}(A,\tau)\) denote the \((N-\tau)\times(N-\tau)\) matrix obtained by puncturing the last \(\tau\) rows and the last \(\tau\) columns of a circulant matrix \(A\). Specifically,
	\begin{equation*}
		{\rm Pu}(A,\tau)=\begin{pmatrix}
			a_0 & a_1 & \cdots & a_{N-\tau - 1}  \\
			a_{N - 1} & a_0 & \cdots & a_{N-\tau - 2}  \\
			\vdots & \vdots & \ddots & \vdots \\
			a_{\tau+1} & a_{\tau+2} & \cdots & a_{0}  \\
		\end{pmatrix},
	\end{equation*}
	where \(1\leq\tau<N\). The square matrix \({\rm Pu}(A,\tau)\) is called a punctured circulant matrix (PUCM). In particular, if \(A\) is a CPM, then \({\rm Pu}(A,\tau)\) is called a punctured circulant permutation matrix (PUCPM).
	
	
	\subsection{LDPC codes}
	An LDPC code $\C$ over $\F_q$ is defined by a sparse parity-check matrix
	\begin{equation}\label{HH1}
		H=\left(h_{i,j}\right)_{\tiny{\substack{1\leq i\leq m_1\\
					1\leq j\leq m_2}}}\in \F_q^{m_1\times m_2}.
	\end{equation}
	The Tanner graph $G$ of $\C$ is a bipartite graph with two sets of nodes: the variable nodes (VNs) and
	check nodes (CNs). The VNs and CNs, denoted by $v_1,v_2,\ldots,v_{m_2}$ and $c_1,c_2,\ldots,c_{m_1}$, respectively. The VNs and CNs correspond to the $m_2$ columns and the $m_1$ rows of $H$. A VN $v_j$ is connected to a CN $c_i$ by an edge if and only if the $h_{i,j}=1$. The degree $d_{v_j}$ of the VN $v_j$ ({\it resp.} $d_{c_i}$ of CNs) is defined as the number of CNs connected to $v_j$ ({\it resp.} the number of VNs connected to $c_i$). A cycle of length $2i$ in $H$ shown in (\ref{HH1}) is defined by $2i$ positions $h_{i,j}=1$ such that:
	\begin{itemize}
		\item [{\rm (1)}]two consecutive positions are obtained by changing alternatively of row or column only;
		\item [{\rm (2)}]all positions are distinct, expect the first and last ones.
	\end{itemize}
	The girth which is the minimum length of cycles in their Tanner graph.
	Since short cycles and the minimum distance directly affect the performance of LDPC codes, when designing LDPC codes, we should try to avoid the occurrence of $4$-cycles and ensure that the minimum distance is not too small.
	
	If the parity-check matrix $H$ can be written as
	\begin{equation}\label{LDPCparitycheck}
		H=\left(
		\begin{array}{cccc}
			A_{1,1} & A_{1,2} & \cdots & A_{1,n} \\
			A_{2,1} & A_{2,2} & \cdots & A_{2,n} \\
			\vdots & \vdots & \ddots & \vdots \\
			A_{r,1} & A_{r,2} & \cdots & A_{r,n} \\
		\end{array}
		\right),
	\end{equation}
	where $A_{i,j}$ is an $N\times N$ CM for $i\in[r]$ and $j\in [n]$, $m_1=Nr$, and $m_2=Nn$, then the corresponding code is a QC LDPC code.
	One of the main reasons why QC LDPC codes are attractive is that they can be efficiently encoded using methods such as those presented in \cite{Myung1,LiZ1}, and can be efficiently decoded using decoding algorithms based on belief propagation \cite{Ksch1} or linear programming \cite{Feld1,Feld2,Tag1,Vont1}.
	
	A $(J,L)$-regular LDPC code is defined as a linear code represented by a parity-check matrix $H$ in which each column has weight $J$ and each row has weight $L$.
	
	If the sub-matrix $A_{i,j}$ of $H$ satisfies $A_{i,j}=P_N^{f(i,j)}$ with $i\in [r]$ and $j\in[n]$,
	then $\C$ with parity-check matrix $H$ is a QC LDPC code composed of circulant permutation matrices (CPM-QC LDPC codes).
	
	For non-binary case,
	define a scaling matrix $S=(s_{i,j})_{i\in [r], j\in [n]}$, where $s_{i,j}\in\F_q^*$.
	If the sub-matrix $A_{i,j}$ of $H$ satisfies $A_{i,j}=s_{i,j}\cdot P_N^{f(i,j)}$, with $i\in [r]$ and $j\in[n]$,
	then $\C$ with parity-check matrix $H$ is a QC LDPC code composed of circulant shift matrices (CSM-QC LDPC codes).
	
	The following theorem is derived from \cite{FOS1}, and it guides us on how to select CPMs to avoid the appearance of 4-cycles in the corresponding Tanner graph.
	\begin{theorem}\label{FOS11}\cite{FOS1}
		Let $H$ be the parity-check matrix of a CPM-QC LDPC code, and let $\mathcal{M}_1=\{f(i,j)|i\in [r], j\in [n]\}$. A necessary and sufficient condition for the Tanner graph representation of $H$ to have a girth at least $6$ is
		\begin{equation*}
			f(i_1,j_1)-f(i_2,j_1)\neq f(i_3,j_2)-f(i_2,j_2) {\rm {~mod~}} N
		\end{equation*}
		for any $\{f(i_1,j_1), f(i_2,j_1), f(i_3,j_2), f(i_2,j_2)\}\in {\mathcal{M}_{1}\choose 4}$.
	\end{theorem}
	It is easy to verify that the above theorem still holds when $H$ is the parity-check matrix of a CSM-QC code.

	\subsection{Golomb ruler}
	Golomb rulers are credited as being ``discovered'' by W. Babcock in 1953 \cite{Bab}.
	Golomb rulers derive their name from Solomon Golomb who was one of the first to study their construction in relation to combinatorics, coding theory and communications, where most of their applications lie.
	A Golomb ruler $GR(M,v)$ is a set of $M$ integers
	$$\{a_1,a_2,\ldots,a_{M}\}\subseteq\{0,1,\ldots,v\}$$
	whose pairwise differences in absolute value take ${M\choose 2}$ distinct values exactly once. By convention, we assume that the sequence $i\mapsto a_i$ is increasing and $a_1=0$, $a_M=v.$
	A collection of integers $a_1,a_2,\ldots,a_{M}$ forms an $N$-modular Golomb ruler if $(a_i-a_j)_{N}$ are nonzero and distinct for distinct order pairs $(i,j)$, $0\leq i\neq j<n$.
	
	\subsection{Vandermonde matrix and Moore matrix}
	A Vandermonde matrix is a matrix where each row comprises terms in geometric progression. Precisely, for \(r\leq m\), an \(r\times m\) Vandermonde matrix \(V(x_1,x_2,\ldots,x_{m};r)\) is defined such that its \((i, j)\)-th entry is given by \(V_{i,j}=x_j^{i - 1}\), with \(i\in [r]\) and \(j\in [m]\).
	
	The determinant of an \(r\times r\) square Vandermonde matrix, denoted as the Vandermonde determinant, can be expressed in the following form:
	\begin{equation*}
		\det(V(x_1,x_2,\ldots,x_{r};r))=\prod_{1\leq i<j\leq r}(x_j - x_i).
	\end{equation*}
	If \(x_1, \ldots, x_{m}\) are \(m\) distinct elements from the multiplicative group \(\mathbb{F}_{q^s}^*\) and \(m \leq q^s - 1\), then by using \(V(x_1,x_2,\ldots,x_{m};r)\) as the parity-check matrix, an \([m,m - r,r + 1]\) RS-type code over \(\mathbb{F}_{q^s}\) can be obtained. If \(x_1, \ldots, x_{m}\) are \(m\) distinct elements from the set \(\{1,x,x^2,\ldots,x^{p - 1}\}\), where \(\{1,x,x^2,\ldots,x^{p - 1}\}\subset R=\mathbb{F}_2[x]/\langle1 + x+\cdots+x^{p - 1}\rangle\), \(p\) is an odd prime, and \(m\leq p - 1\), then by using \(V(x_1,x_2,\ldots,x_{m};r)\) as the parity-check matrix, a BR-type code over \(R\) can be derived. This BR-type code is a \([m,m - r,r + 1;p - 1]\) array code over \(\mathbb{F}_q\).
	
	A Moore matrix \(M(\alpha_1,\alpha_2,\ldots,\alpha_{m};r)\) is a matrix defined over a finite field \(\mathbb{F}_{q^{t_1}}\), where \(t_1\) and \(m\) are positive integers that satisfy the inequality \(r\leq m\leq t_1\). Additionally, for each \(i\in [m]\), the element \(\alpha_i\) belongs to the multiplicative group \(\mathbb{F}_{q^{t_1}}^*\).  Its columns are obtained by applying successive powers of the Frobenius automorphism. As such, it is an \(r\times m\) matrix given by
	\begin{equation}\label{Moore}
		M(\alpha_1,\alpha_2,\ldots,\alpha_{m};r)=\left(\alpha_{j}^{q^{i-1}}\right)_{r\times m}
	\end{equation}
	for \(i\in [r]\) and \(j\in [m]\).
	The Moore determinant of an \(r\times r\) square Moore matrix can be expressed as:
	\begin{equation}\label{Mooredet1}
		\det(M(\alpha_1,\alpha_2,\ldots,\alpha_{r};r))=\prod_{\mathbf{c}}(c_1\alpha_1+\cdots+c_r\alpha_r),
	\end{equation}
	where \(\mathbf{c}=(c_1,c_2,\ldots,c_r)\) traverses a complete set of direction vectors in \(\mathbb{F}_q^r\), with the specific condition that the last nonzero entry is equal to \(1\).
	\section{Block MDS LDPC codes form CMs with column weight $1$}
	In this section, we provide two schemes for constructing block MDS LDPC codes. The first scheme, which is based on CPMs, is a scheme that requires a puncturing operation. Its advantage is that it can be applied to any finite field. The second scheme, which is based on CSMs, is a scheme that does not require a puncturing operation, but it is only applicable to non-binary cases.
	\subsection{The scheme where puncturing is required}
	In this subsection, we construct block MDS array codes by appropriately puncturing array codes over $\F_q$ with the Vandermonde form $V\left(I_N,P_N,\ldots,P_{N}^{N-1};r\right)$, where $r\leq N-1$ and
	\begin{equation}\label{Van1}
		V\left(I_N,P_N,\ldots,P_{N}^{N-1};r\right)=\left(
		\begin{array}{cccc}
			I_N & I_N & \cdots & I_N \\
			I_N & P_N & \cdots & P_N^{N-1} \\
			\vdots & \vdots & \ddots & \vdots \\
			I_N & P_N^{r-1} & \cdots & P_N^{(r-1)(N-1)} \\
		\end{array}
		\right).
	\end{equation}
	It is not difficult to check that $V\left(I_N,P_N,\ldots,P_{N}^{N-1};r\right)$ has rank $Nr-r+1$.
	
	Lv {\it et al.} \cite{Fang1,Lv1} presented the following theorem to guide us on how to select polynomials and perform puncturing to obtain MDS array codes.
	\begin{lemma}\label{Lv1}\cite{Lv1}
		Let $r$ and $m$ be positive integers with $r<m$.
		Let $a_1(x),a_2(x),\ldots,a_{m}(x)\in \frac{\F_q[x]}{\langle x^N-1\rangle}$, and let $A_i$ be the associated circulant matrices of $a_i(x)$ for $i\in[m]$. Assume that $\gcd\left(a_i(x)-a_j(x),x^N-1\right)=g(x)$ for all $1\leq i<j\leq m$, where $\deg(g(x))=\tau$. Then the code with parity-check matrix
		\begin{equation*}
			H^{\prime}=\left(
			\begin{array}{cccc}
				I_{N-\tau} & I_{N-\tau} & \cdots & I_{N-\tau}  \\
				{\rm Pu}(A_1,\tau) & {\rm Pu}(A_2,\tau) & \cdots & {\rm Pu}(A_l,\tau)  \\
				{\rm Pu}(A_1^2,\tau) & {\rm Pu}(A_2^2,\tau) & \cdots & {\rm Pu}(A_l^2,\tau)  \\
				\vdots & \vdots & \ddots & \vdots \\
				{\rm Pu}(A_1^{r-1},\tau) & {\rm Pu}(A_2^{r-1},\tau) & \cdots & {\rm Pu}(A_l^{r-1},\tau)  \\
			\end{array}
			\right)
		\end{equation*}
		is an MDS array code with sub-packetization level $N-\tau$.
	\end{lemma}
	If \(N\) is an odd prime, Fan \cite{Fan1} proved that for the parity-check matrix $H$ as in (\ref{Van1}), the corresponding Tanner graph has no 4-cycles.
	The code constructed by the following theorem is both an MDS array code and free of $4$-cycles.
	\begin{theorem}\label{thmfirst}
		Assume that $q$ is a primitive root modulo $N$, where $N$ is an odd prime. Then the code with parity-check matrix
		\begin{equation}\label{VVV1}
			H^{\prime\prime}=\left(
			\begin{array}{cccc}
				I_{N-1} & I_{N-1} & \cdots & I_{N-1} \\
				I_{N-1} & {\rm Pu}(P_{N},1) & \cdots & {\rm Pu}(P_N^{N-1},1) \\
				\vdots & \vdots & \ddots & \vdots \\
				I_{N-1} & {\rm Pu}(P_N^{J-1},1) & \cdots & {\rm Pu}(P_N^{(J-1)(N-1)},1) \\
			\end{array}
			\right)
		\end{equation}
		is both an MDS array code over $\F_q$ with sub-packetization level $N-1$ and free of $4$-cycles.
	\end{theorem}
	\begin{proof}
		It can be readily verified that when \(N\) is an odd prime, the puncturing operation from \(H^{\prime}\) to obtain \(H^{\prime\prime}\) will not lead to the emergence of $4$-cycles. Subsequently, we confirm that \(H^{\prime\prime}\) is an MDS array matrix with a sub-packetization level of \(N - 1\). The polynomials corresponding to \(I_N, P_N,\ldots, P_N^{N - 1}\) in \(\frac{\mathbb{F}_q[x]}{\langle x^n - 1\rangle}\) are \(1,x,\ldots,x^{N - 1}\) respectively. Given that \(q\) is a primitive root modulo \(N\), the irreducible factors of \(x^N - 1\) are \(x - 1\) and \(1 + x+\cdots+x^{N - 1}\). In accordance with Lemma \ref{Lv1}, the greatest common divisor of the difference between any two monomials \(x^{i_1}\) and \(x^{i_2}\) and \(x^N - 1\) can only be \(x - 1\). Thus, after puncturing the last row and the last column of each block of \(H^{\prime}\), we obtain the aforementioned results.
	\end{proof}
	Artin's conjecture \cite{Mor1} posits that for any integer \(a\) that is neither \(\pm1\) nor a perfect square, there are infinitely many prime numbers \(p\) for which \(a\) is a primitive root modulo \(p\). In 1967, Hooley \cite{Hoo1} established the truth of this conjecture under the assumption of the Generalized Riemann Hypothesis. At present, the entire number theory community holds the view that the Generalized Riemann Hypothesis is correct \cite{Con1}. Consequently, we assert that a family of block MDS LDPC codes can be constructed in accordance with Theorem \ref{thmfirst}.
	\begin{example}\label{ex7}
		For high rate and moderate length (say up to about $5000$), the LDPC codes with the parity-check matrix as in (\ref{Van1}) perform as well as the best comparable randomly constructed regular LDPC codes presented in \cite{MacK0}.
		We choose \(N = 67\). It is straightforward to verify that \(N\) is a prime number and \(2\) is a primitive root of \(N\). When the parity-check matrix is as shown in (\ref{Van1}), we obtain a \((4489, 4288)\) LDPC code \(\mathcal{C}_1\). By puncturing according to Theorem \ref{thmfirst}, we get a \((4422, 4224)\) block MDS LDPC code \(\mathcal{C}_2\). Adopting the degree distributions of $\C_1$ and \(\mathcal{C}_2\), the bit-error rate (BER) and block-error rate (BLER) performances of an LDPC code \(\mathcal{C}_3\) and \(\mathcal{C}_4\) randomly generated by the Progressive Edge-Growth (PEG) algorithm are presented in Figure \ref{f1}. Their rates are all 0.955.  We compare the decoding performance of the following fours codes under AWGN channel, with sum-product decoder where max iteration is set to $50$. As can be observed from Figure \ref{f1}, there is almost no degradation in the BER and BLER performance of \(\mathcal{C}_2\).
		\begin{figure}[H]
			\centering
			\includegraphics[width=0.8\textwidth]{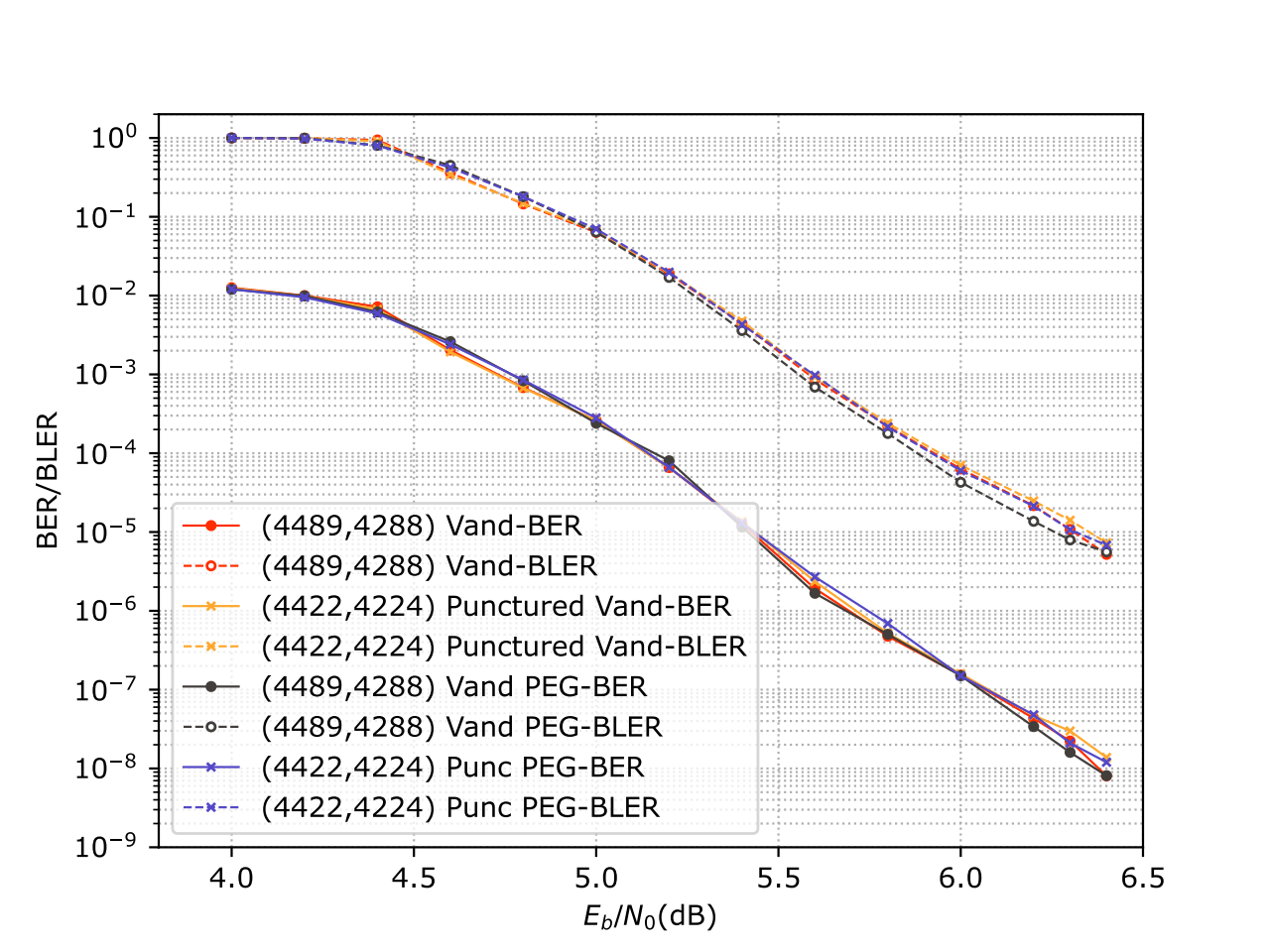}
			\caption{The BER and BLER performances of four LDPC codes}
			\label{f1}
		\end{figure}

	\end{example}
	Under the condition of solely preventing the existence of $4$-cycles in the Tanner graph, the block MDS LDPC codes constructed in accordance with Theorem \ref{thmfirst} possess a code length of \(N(N - 1)\) and a code rate of \((N - r)/N\). Here, \(N\) is an odd prime number for which $2$ serves as a primitive root, and \(r\) represents the redundancy of the array code. Fan \cite{Fan1} proposed a strategy for eliminating short cycles, albeit at the expense of sacrificing the code length.
	If the goal is to eliminate the $6$-cycles in \(V(I_N, P_N,\ldots, P_{N}^{N - 1})\), we can utilize the matrix \(V(P_N^{i_0}, P_N^{i_1},\ldots, P_{N}^{i_{s - 1}})\) as the parity-check matrix. In this case, \(1\leq s - 1\leq N - 1\), and \(M_2=\{i_0, i_1,\ldots, i_{s - 1}\}\subset \mathbb{Z}_N\) is such that for any \(\{k_1, k_2, k_3\}\in {\mathbb{Z}_s \choose 3}\), the inequality \(N\nmid (k_1 - k_3)i_{1}+(k_2 - k_1)i_{2}+(k_3 - k_2)i_{3}\) holds. Based on this strategy, the following example can be derived.
	\begin{example}\label{ex8}
		We adopt the parity-check matrices of $\mathcal{C}_1$ and $\mathcal{C}_2$ in Example \ref{ex7}, denoted as $H_1$ and $H_2$ respectively. To eliminate the $6$-cycles in their corresponding Tanner graphs, we only retain the columns (here referring to block-columns) of $H_1$ and $H_2$ on the set $M_2$, where $M_2=\{23,29,12,55,43,28,37,38,10,65,56,41\}\subset[67]$. It is easy to verify that $M_2$ satisfies the above-mentioned condition for eliminating $6$-cycles. Therefore, the girths of the LDPC codes $\mathcal{C}_5$ and $\mathcal{C}_6$ with $H_1^{\prime}$ and $H_2^{\prime}$ as their parity-check matrices are greater than or equal to $8$ (we verified that the girth is $8$ using Python). The parameters of $\mathcal{C}_5$ and $\mathcal{C}_6$ are $(871,670)$ and $(858,660)$ respectively, and both have a code rate of $0.769$.
		
		For comparison, we select the first $13$ columns of $H_1$ and $H_2$ to obtain new parity-check matrices $H_1^{\prime\prime}$ and $H_2^{\prime\prime}$, and the corresponding LDPC codes are $\mathcal{C}_7$ and $\mathcal{C}_8$. The parameters of $\mathcal{C}_7$ and $\mathcal{C}_8$ are $(871,670)$ and $(858,660)$ respectively, and both have a code rate of $0.769$. Their girths are both $6$. As can be seen from Figure \ref{f2}, the performance of $\mathcal{C}_5$ and $\mathcal{C}_6$ is superior to that of $\mathcal{C}_7$ and $\mathcal{C}_8$. Therefore, increasing the girth can indeed improve the performance of LDPC codes under this construction.
		\begin{figure}[H]
			\centering
			\includegraphics[width=0.8\textwidth]{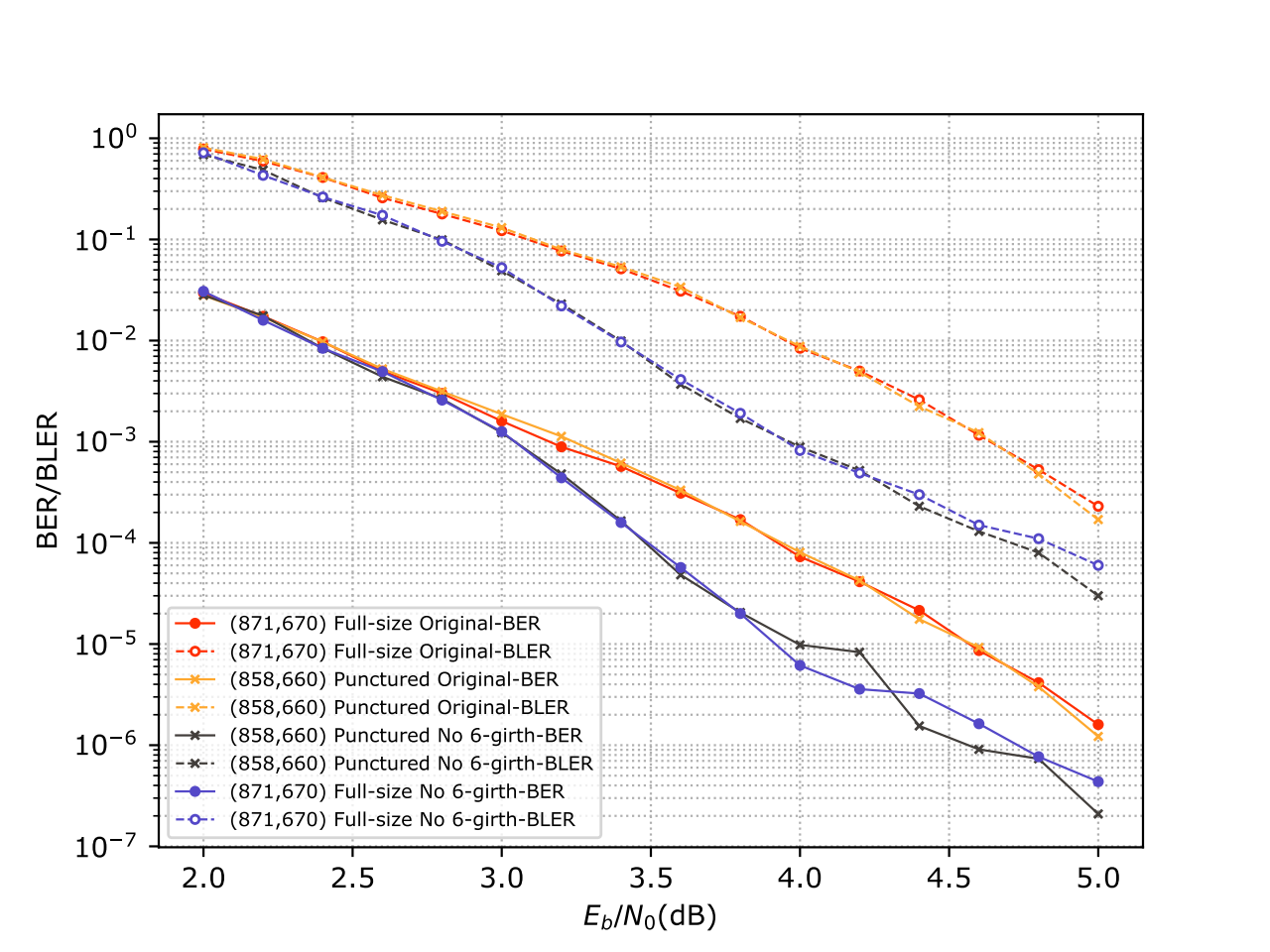}
			\caption{The BER and BLER performances of eight LDPC codes}
			\label{f2}
		\end{figure}
	\end{example}
	\subsection{The scheme where puncturing is not required}
	As Fossorier pointed out, \(H\) in (\ref{H1}) cannot have full row rank. Therefore, there is no binary CPM-QC LDPC code without $4$-cycles.
	Tauz {\it et al.}\cite{Tau1} utilized the generalized Schur's formula (i.e., Lemma \ref{BlockVan}) and the necessary and sufficient condition (i.e., Theorem \ref{FOS11}) provided by Fossorier \cite{FOS1} to avoid the occurrence of $4$-cycles. By combining with the scaling matrix $S$, they could theoretically construct non-binary block MDS LDPC codes. As can be seen from Lemma \ref{BlockVan}, it is not easy to implement the generalized Schur's formula to determine whether an array code is an MDS array code. In this subsection, we present a more easily implementable method for constructing non-binary block MDS LDPC codes.
	
	Define the scaling matrix $S=(s_{t_j}^{i-1})_{i\in[r], j\in [m]}$ over $\F_q$, where $\{t_1\ldots,t_m\}\in {[N]\choose m}$ and $(s_{t_1},\ldots,s_{t_{m}})\in(\mathbb{F}_q^*)^m$. Consider the Vandermonde matrix $V\left(s_{t_1}\cdot P_N^{t_1},\ldots,s_{t_m}\cdot P_{N}^{t_m};r\right)$, where $m>r$ and
	{\small\begin{equation}\label{Van2}
			V\left(s_{t_1}\cdot P_N^{t_1},\ldots,s_{t_m}\cdot P_{N}^{t_m};r\right)=\left(
			\begin{array}{cccc}
				I_N & I_N & \cdots & I_N \\
				s_{t_1}\cdot P_N^{t_1} & s_{t_2}\cdot P_N^{t_2} & \cdots & s_{t_m}\cdot P_N^{t_m} \\
				\vdots & \vdots & \ddots & \vdots \\
				s_{t_1}^{r-1}\cdot P_N^{t_1\cdot(r-1)} & s_{t_2}^{r-1}\cdot P_N^{t_2\cdot(r-1)} & \cdots & s_{t_m}^{r-1}\cdot P_N^{t_m\cdot(r-1)} \\
			\end{array}
			\right).
	\end{equation}}
	\begin{theorem}\label{Taurefine}
		Let $m$ and $N$ be positive integers with $N\geq m>r$ and let $\F_q$ be a non-binary field with $q-1\geq m$.
		Suppose there exist $S_1=\{t_1,\ldots,t_m\}\in {[N]\choose m}$ and $S_2=\{s_{t_1},\ldots,s_{t_{m}}\}\in {\mathbb{F}_q^*\choose m}$ such that $\gcd\left(s_{t_i}x^{t_i}-s_{t_j}x^{t_j},x^N - 1\right)=1$ for all $\{i,j\}\in{[m]\choose 2}$. Then the code with the parity-check matrix as in (\ref{Van2}) is both an MDS array code $\C$ over $\F_q$ with parameters $(m,m-r,d_A(\C)=r+1;N)$ and free of $4$-cycles.
	\end{theorem}
	\begin{proof}
		The proof of this theorem can be directly obtained from Lemma \ref{Lv1}. It is worth noting that we need to restrict that $\mathbb{F}_q$ is not a binary field and $q-1\geq m$. Otherwise, $\gcd(s_{t_i}x^{t_i}-s_{t_j}x^{t_j},x^N - 1)$ cannot be equal to $1$.
	\end{proof}
	Similar to Theorem \ref{thmfirst}, we can restrict that $q$ is a primitive root modulo $N$, where $N$ is an odd prime, so as to reduce the difficulty of searching for $S_1$ and $S_2$.
	\begin{example}\label{ex10}
		Assume \( m =N = 67 \) and $r=3$. Let \(\mathbb{F}_{128} = \left\{ \sum_{i=0}^{6} a_i \alpha^{i} \mid a_i \in \mathbb{F}_2 \right\}\), where \(\alpha\) is a root of the polynomial \(x^7 + x + 1\). It is easy to verify that 128 is a primitive root modulo 67. For convenience, we use the binary representation of a number \(a\) to denote the corresponding element in \(\mathbb{F}_{128}\). For example, since \(a=15 = 1 + 1 \cdot 2 + 1 \cdot 2^2 + 1 \cdot 2^3\), the number $a$ represents the element \(1 + \alpha + \alpha^2 + \alpha^3 \).
		Select the sets \(S_1 = [67]\) and \(S_2 = \{86, 15, 44, 57, 55, 37, 21, 98, 10, 79, 51, 71, 16, 76, 28, 124, 36, 13, 61, 24, 58, 39, 101, 31, 29, 9,\\ 113, 22, 75, 60, 100, 45, 116, 46, 26, 68, 5, 84, 53, 59, 110, 92, 97, 3, 94, 105, 33, 65, 95, 32, 99,\\ 126, 89, 12, 119, 62, 120, 27, 17, 109, 102, 52, 115, 11, 127, 104, 106\} \subset \mathbb{F}_{128}^*.\)
		It is readily verified that the sets \( S_1 \) and \( S_2 \) satisfy the conditions in Theorem \ref{Taurefine}. Thus, we obtain a block MDS LDPC code \( \mathcal{C}_9 \), whose parameters as an array code are \( (67, 64, d_A(\C_9)=4; 67) \). Figure \ref{fig3} shows the decoding performance curves of \( \mathcal{C}_9 \) over $\F_{128}$ under the AWGN channel.
		\begin{figure}[H]
			\centering
			\includegraphics[width=0.8\textwidth]{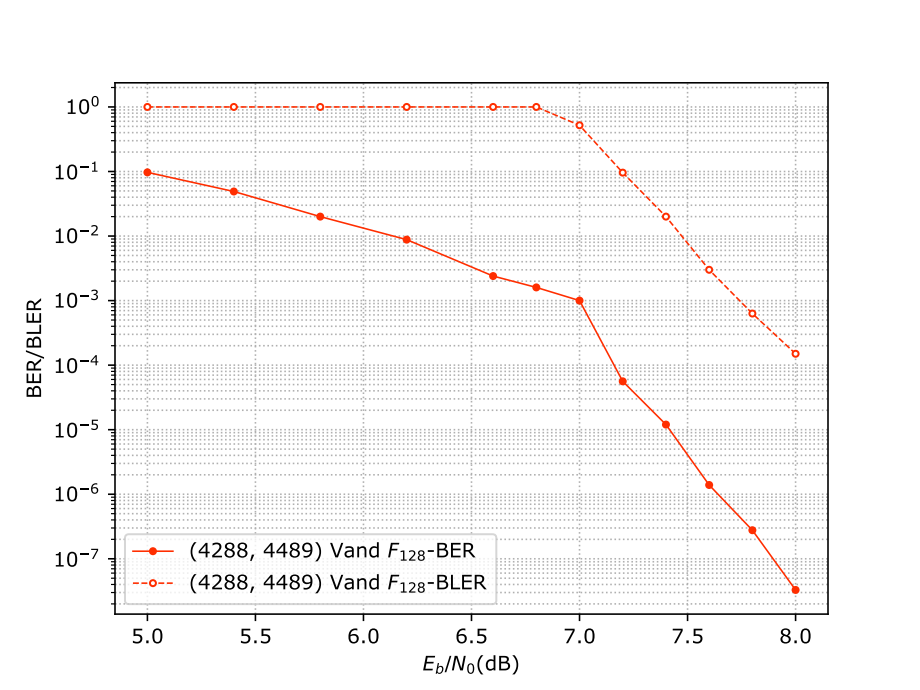}
			\caption{The BER and BLER performances of $\C_9$ over $\F_{128}$}
			\label{fig3}
		\end{figure}
	\end{example}
	\section{Block MDS LDPC codes form CMs with column weight $>1$}
	In this section, we delve into the construction of block MDS LDPC codes over $\F_q$ based on punctured CM$(t)$s, where CM$(t)$ denotes a circulant matrix with a column weight (or row weight) of \(t\). In this section, we default that the characteristic of $\mathbb{F}_q$ is $2$.
	
	Smarandache and Vontobel \cite{Smar1,Smar2} studied CM$(t)$-QC LDPC codes and gave upper bounds on their minimum distance, where $t$ can take values of $0, 1, 2, 3$ (and sometimes larger). He {\it et al.} \cite{He1} proposed a PEG algorithm applicable to CM$(t)$-QC LDPC codes.
	From the perspective of the fundamental criterion for eliminating the presence of $4$-cycles in the Tanner graph, using circulant matrices with a column weight greater than $1$ as a basic building block does not seem to be a prudent choice.  However, the results of our simulation experiments reveal that the LDPC codes constructed with PuCM$(t)$ as the block exhibit excellent performance. This compels us to give due consideration to this situation.
	
	Although, with the assistance of Theorem \ref{thmfirst}, an MDS array code can be obtained by utilizing a Vandermonde matrix as the parity-check matrix, when the column weight of block \(A_{i}\) is greater than $1$, \(A_{i}^s\) has a high probability of becoming highly dense. As a result, it becomes impossible to avoid the formation of $4$-cycles for some values of \(s\). For example, consider the polynomial \(a(x)\in\frac{\mathbb{F}_2[x]}{\langle x^N - 1\rangle}\), and its corresponding \(7\times7\) matrix is given by:
	\[A = \begin{bmatrix}
		1 & 1 & 0 & 0 & 0 & 0 & 0 \\
		0 & 1 & 1 & 0 & 0 & 0 & 0 \\
		0 & 0 & 1 & 1 & 0 & 0 & 0 \\
		0 & 0 & 0 & 1 & 1 & 0 & 0 \\
		0 & 0 & 0 & 0 & 1 & 1 & 0 \\
		0 & 0 & 0 & 0 & 0 & 1 & 1 \\
		1 & 0 & 0 & 0 & 0 & 0 & 1
	\end{bmatrix}.\]
	Then,
	\[A^3 = \begin{bmatrix}
		1 & 1 & 1 & 1 & 0 & 0 & 0 \\
		0 & 1 & 1 & 1 & 1 & 0 & 0 \\
		0 & 0 & 1 & 1 & 1 & 1 & 0 \\
		0 & 0 & 0 & 1 & 1 & 1 & 1 \\
		1 & 0 & 0 & 0 & 1 & 1 & 1 \\
		1 & 1 & 0 & 0 & 0 & 1 & 1 \\
		1 & 1 & 1 & 0 & 0 & 0 & 1
	\end{bmatrix}.\]
	At this juncture, $4$-cycles are present in \(A^3\). For \(a(x)\in \frac{\mathbb{F}_2[x]}{\langle x^N-1\rangle}\), we can effectively circumvent the issue of having an excessive number of non-zero coefficients in \(a^s(x)\) by choosing a parity-check matrix in Moore form.
	
	We will carry out a detailed case-by-case analysis on the methods to prevent the appearance of $4$-cycles in the Tanner graph of \(H\). Suppose that \(a(x)=x^{i_1}+x^{i_2}+\cdots + x^{i_t}\in\frac{\mathbb{F}_2[x]}{\langle x^N-1\rangle}\), where \(i_1 < i_2<\cdots < i_t\in\mathbb{Z}_N\). For the sake of simplicity and convenience, we define \({\rm Index}(a(x))=\{i_1,i_2,\ldots,i_t\}\) and \(L_j(a(x))=\{(k,k + i_j)|k\in [N]\}\) for \(j\in[t]\). We also define \(\Delta(M)=\left\{(i_{j_1}-i_{j_2})_N|\{i_{j_1},i_{j_2}\}\in{M\choose 2}\right\}\), where \(M\subset\mathbb{Z}_N\).
	Furthermore, we define the addition and subtraction operations between sets as \(A + B=\{(a + b)_N\mid a\in A, b\in B\}\) and \(A - B=\{(a - b)_N\mid a\in A, b\in B\}\) respectively, where \(A\) and \(B\) are subsets of \(\mathbb{Z}_N\).
	\subsection*{Absence of $4$-cycles in a single CM$(t)$}
	When \(t = 1\), it is impossible for any cycles to appear in a single CM\((t)\). However, when \(t>1\), cycles will occur in a single CM\((t)\).
	\begin{lemma}\label{lemm11}
		Let \(N\) be an odd integer. Suppose \(t=2\) and \(a(x)=x^{i_1}+x^{i_2}\in \frac{\F_q[x]}{\langle x^N-1\rangle}\), where \(i_1<i_2\in\mathbb{Z}_N\). Then the girth \(g\) of the corresponding matrix CM\((t)\) is \(g = \frac{2N}{\gcd(N,i_2 - i_1)}>4\). Suppose \(t>2\) and \(a(x)=x^{i_1}+x^{i_2}+\cdots + x^{i_t}\in \frac{\F_q[x]}{\langle x^N-1\rangle}\), where \(i_1 < i_2<\cdots < i_t\in\mathbb{Z}_N\). Then the girth \(g\) of the matrix CM\((t)\) is greater than 4 if \(\Delta ({\rm Index}(a(x))\) is a set.
	\end{lemma}
	\begin{proof}
		We commence by examining the case where \(t = 2\). Without loss of generality, we select the node \((1,1 + i_2)\) as the initial node. Traversing along the vertical path, we arrive at the node \((1 + i_2 - i_1,1 + i_2)\). Subsequently, following the horizontal path, we encounter the node \((1 + i_2 - i_1,1 + 2i_2 - i_1)\). Then, traversing vertically once again, we reach the node \((1 + 2i_2 - 2i_1,1 + 2i_2 - i_1)\).
		If the node \((1,1 + i_2)\) lies on a $4$-cycle, then \(N\mid 2(i_2 - i_1)\). This contradicts the fact that \(N\) is an odd integer. Hence, we conclude that \(g>4\).
		Proceeding from the node \((1 + 2i_2 - 2i_1,1 + 2i_2 - i_1)\), after traversing another horizontal and vertical path, we will reach the node \((1 + 3i_2 - 3i_1,1 + 3i_2 - 2i_1)\). If the node \((1,1 + i_2)\) lies on a $6$-cycle, then \(N\mid 3(i_2 - i_1)\). In the case where \(N\nmid 3(i_2 - i_1)\), by continuing the same traversal steps, we can infer that the node \((1,1 + i_2)\) will lie on a \(\frac{2N}{\gcd(N,i_2 - i_1)}\)-cycle. This is because \(N\) must divide \(\gcd(N,i_2 - i_1)\cdot (i_2 - i_1)\).
		
		When \(t>2\), the four nodes that form a $4$-cycle must be sourced from either three sets \(L_{j_1}(a(x)), L_{j_2}(a(x)), L_{j_3}(a(x))\) or four sets \(L_{j_1}(a(x)), L_{j_2}(a(x)), L_{j_3}(a(x)), L_{j_4}(a(x))\) for some \(\{j_1,j_2,j_3,j_4\}\in\binom{\text{Index}(a(x))}{4}\). These scenarios correspond to the following two diagrams respectively:
		\[
		\begin{minipage}{\textwidth}
			\footnotesize
			\begin{equation*}
				\xymatrix{
					\left(1,1 + i_{j_1}\right)\in  L_{j_1}(a(x))\ar[d] & \left(1+\underline{(i_{j_1}-i_{j_2}+i_{j_3}-i_{j_2})_N},1+(i_{j_1}-i_{j_2}+i_{j_3})_N\right)\in L_{j_2}(a(x)) \ar[l] \\
					\left(1+(i_{j_1}-i_{j_2})_N,1 + i_{j_1}\right)\in L_{j_2}(a(x)) \ar[r] & \left(1+(i_{j_1}-i_{j_2})_N,1+(i_{j_1}-i_{j_2}+i_{j_3})_N\right)\in L_{j_3}(a(x)) \ar[u]
				}
			\end{equation*}
			\begin{equation*}
				\xymatrix{
					\left(1,1 + i_{j_1}\right)\in  L_{j_1}(a(x))\ar[d] & \left(1+\underline{(i_{j_1}-i_{j_2}+i_{j_3}-i_{j_4})_N},1+(i_{j_1}-i_{j_2}+i_{j_3})_N\right)\in L_{j_4}(a(x)) \ar[l] \\
					\left(1+(i_{j_1}-i_{j_2})_N,1 + i_{j_1}\right)\in L_{j_2}(a(x)) \ar[r] & \left(1+(i_{j_1}-i_{j_2})_N,1+(i_{j_1}-i_{j_2}+i_{j_3})_N\right)\in L_{j_3}(a(x)) \ar[u]
				}
			\end{equation*}
		\end{minipage}
		\]
		To prevent the occurrence of $4$-cycles, we only need to ensure that the underlined elements in the two diagrams are never zero. If \(\Delta(\text{Index}(a(x)))\) is a set, it implies that \(\text{Index}(a(x))\) effectively forms an \(N\)-modular Golomb ruler. Consequently, this serves as a sufficient condition for avoiding the appearance of $4$-cycles within a single CM\((t)\).
	\end{proof}
	\subsection*{Absence of $4$-cycles between two CM$(t)$s}
	Suppose \(H\) is in the form shown in (\ref{LDPCparitycheck}). There may be $4$-cycles in $\left(
	\begin{array}{c}
		A_{i_1,j_1} \\
		A_{i_2,j_1} \\
	\end{array}
	\right)$
	or \(\left(A_{i_1,j_1},A_{i_1,j_2}\right)\). The following lemma guides us to avoid $4$-cycles between two CM$(t)$s.
	\begin{lemma}\label{lemm12}
		Let \(N\) be an odd integer. Suppose \(t\geq 2\), \(a_1(x)=\sum_{j=1}^t x^{i_{1,j}}\), and \(a_2(x)=\sum_{j=1}^t x^{i_{2,j}}\) where \(i_{1,j}, i_{2,j}\in\mathbb{Z}_N\) for $j\in [t]$. Then the girth \(g\) of the matrix \((A_1,A_2)\) is greater than $4$ if both \(\Delta ({\rm Index}(a_1(x)))\) and \(\Delta ({\rm Index}(a_2(x)))\) are sets and $$\Delta ({\rm Index}(a_1(x)))\cap \Delta ({\rm Index}(a_2(x)))=\varnothing,$$ where \(A_1\) and \(A_2\) are the corresponding CM$(t)$s of \(a_1(x)\) and \(a_2(x)\) respectively.
	\end{lemma}
	\begin{proof}
		The condition ``both \(\Delta ({\rm Index}(a_1(x)))\) and \(\Delta ({\rm Index}(a_2(x)))\) are sets'' is to ensure that there are no $4$-cycles in a single CM\((t)\) by Lemma \ref{lemm11}.
		The four nodes that form a $4$-cycle must have two nodes coming from \(L_{j_1}(a_1(x))\) and \(L_{j_2}(a_1(x))\), and the remaining two nodes come from \(L_{j_3}(a_2(x))\) and \(L_{j_4}(a_2(x))\) for some \(\{j_1,j_2\}\in {{\rm Index}(a_1(x))\choose 2}\) and \(\{j_3,j_4\}\in {{\rm Index}(a_2(x))\choose 2}\). The corresponding diagram is as follows:
		{\tiny \begin{equation*}
				\xymatrix{
					\left(1,1+i_{1,j_1}\right)\in  L_{j_1}(a_1(x))\ar[d] & \left(1+\underline{(i_{1,j_1}-i_{1,j_2}+i_{2,j_3}-i_{2,j_4})_N},1+(i_{1,j_1}-i_{1,j_2}+i_{2,j_3})_N\right)\in L_{j_4}(a_2(x)) \ar[l] \\
					\left(1+(i_{1,j_1}-i_{1,j_2})_N,1+i_{1,j_1}\right)\in L_{j_2}(a_1(x)) \ar[r] & \left(1+(i_{1,j_1}-i_{1,j_2})_N,1+(i_{1,j_1}-i_{1,j_2}+i_{j_3})_N\right)\in L_{j_3}(a_2(x)) \ar[u]
				}
		\end{equation*}}
		To avoid the occurrence of 4-cycles, it is sufficient to ensure that the underlined elements in the above diagram are never zero. Since \(\Delta ({\rm Index}(a_1(x)))\cap \Delta ({\rm Index}(a_2(x))) = \varnothing\), the underlined elements are definitely not zero.
	\end{proof}
	\subsection*{Absence of $4$-cycles among four CM$(t)$s}
	There may be $4$-cycles in $\left(
	\begin{array}{cc}
		A_{i_1,j_1} & A_{i_1,j_2}\\
		A_{i_2,j_1} & A_{i_2,j_2}\\
	\end{array}
	\right)$. The following lemma helps us ensure the absence of 4-cycles in this case.
	\begin{lemma}\label{suff1}
		Let \(N\) be an odd integer. Suppose \(t\geq 2\), \(a_{1,1}(x)=\sum_{j=1}^t x^{i_{1,1,j}}\), \(a_{1,2}(x)=\sum_{j=1}^t x^{i_{1,2,j}}\), \(a_{2,1}(x)=\sum_{j=1}^t x^{i_{2,1,j}}\), and \(a_{2,2}(x)=\sum_{j=1}^t x^{i_{2,2,j}}\), where \(i_{1,1,j}, i_{1,2,j}, i_{2,1,j}, i_{2,2,j}\in\mathbb{Z}_N\) for $j\in [t]$. Then the girth \(g\) of the matrix $\left(
		\begin{array}{cc}
			A_{1,1} & A_{1,2}\\
			A_{2,1} & A_{2,2}\\
		\end{array}
		\right)$ exceeds $4$ if the following three conditions hold:
		\begin{itemize}
			\item [{\rm (i)}]Each of the sets \(\Delta ({\rm Index}(a_{1,1}(x)))\), \(\Delta ({\rm Index}(a_{1,2}(x)))\), \(\Delta ({\rm Index}(a_{2,1}(x)))\), and\\ \(\Delta ({\rm Index}(a_{2,2}(x)))\) is a set.
			\item [{\rm (ii)}]Vertical direction:
			$$\Delta ({\rm Index}(a_{1,1}(x)))\cap \Delta ({\rm Index}(a_{2,1}(x)))=\Delta ({\rm Index}(a_{1,2}(x)))\cap \Delta ({\rm Index}(a_{2,2}(x)))=\varnothing.$$
			Horizontal direction:
			$$\Delta ({\rm Index}(a_{1,1}(x)))\cap \Delta ({\rm Index}(a_{1,2}(x)))=\Delta ({\rm Index}(a_{2,1}(x)))\cap \Delta ({\rm Index}(a_{2,2}(x)))=\varnothing.$$
			\item [{\rm (iii)}]\(0\notin({\rm Index}(a_{1,1}(x))+ {\rm Index}(a_{2,2}(x)))- ({\rm Index}(a_{1,2}(x))+ {\rm Index}(a_{2,1}(x)))\),
		\end{itemize}
		where \(A_{1,1}\), \(A_{1,2}\), \(A_{2,1}\), and \(A_{2,2}\) are the corresponding CM$(t)$s of \(a_{1,1}(x)\), \(a_{1,2}(x)\), \(a_{2,1}(x)\), and \(a_{2,2}(x)\) respectively.
	\end{lemma}
	\begin{proof}
		Conditions (i) and (ii) are designed to guarantee the non-existence of $4$-cycles, both within a single CM\((t)\) and between two CM\((t)\)s by Lemma $\ref{lemm11}$ and Lemma \ref{lemm12}. Consider a $4$-cycle formed by four nodes. Necessarily, one node must be from \(L_{j_1}(a_{1,1}(x))\), one from \(L_{j_2}(a_{2,1}(x))\), one from \(L_{j_3}(a_{2,2}(x))\), and one from \(L_{j_4}(a_{1,2}(x))\), where \(j_1\in {\rm Index}(a_{1,1}(x))\), \(j_2\in {\rm Index}(a_{2,1}(x))\), \(j_3\in {\rm Index}(a_{2,2}(x))\), and \(j_4\in {\rm Index}(a_{1,2}(x))\). The corresponding graphical representation is given by:
		\begin{equation*}
			\xymatrix{
				\left(1,f_1\right)\in  L_{j_1}(a_{1,1}(x))\ar[d] & \left(\underline{f_4},f_3\right)\in L_{j_4}(a_{1,2}(x)) \ar[l] \\
				\left(f_2,f_1\right)\in L_{j_2}(a_{2,1}(x)) \ar[r] & \left(f_2,f_3\right)\in L_{j_3}(a_{2,2}(x)) \ar[u]
			}
		\end{equation*}
		Here, \(f_1 = 1 + i_{1,1,j_1}\), \(f_2=1+(i_{1,1,j_1}-i_{2,1,j_2})_N\), \(f_3=1+(i_{1,1,j_1}-i_{2,1,j_2}+i_{2,2,j_3})_N\), and \(f_4=1+(i_{1,1,j_1}-i_{2,1,j_2}+i_{2,2,j_3}-i_{1,2,j_4})_N=1+(i_{1,1,j_1}+i_{2,2,j_3}-(i_{2,1,j_2}+i_{1,2,j_4}))_N\). It can be readily verified that condition (iii) serves as a sufficient condition to ensure that \(f_4\) is never equal to $1$.
	\end{proof}
	By adhering to the three constraints given in Lemma 12, we can select appropriate polynomials \(a_1(x), a_2(x), \ldots, a_m(x) \in \frac{\mathbb{F}_2[x]}{\langle x^n - 1\rangle}\) to construct the Moore matrix \(M(a_1(x), a_2(x),\\ \ldots, a_m(x))\), in which the corresponding Tanner graph is absent of $4$-cycles.
	
	Our next task is to construct MDS array codes with Moore form.
	A similar work is that Ye and Barg \cite[Lemma 14]{Ye1} gave a calculation formula for the Block Vandermonde determinant when constructing MDS array codes. Equation (\ref{Mooredet1}) is only applicable when the entries in the determinant are elements $\alpha_i$ in the finite field $\mathbb{F}_{q^{t_1}}$. When these $\alpha_i$ are replaced by CM$(t)$s, we also need to give a calculation formula for the Moore determinant of CM$(t)$s.
	
	\begin{lemma}\label{CMsMoo}
		(CMs Moore determinant) Let \(A_1, A_2, \ldots, A_r\) be \(r\) \(N\times N\) CM$(t)$s over \(\mathbb{F}_q\). Then
		\begin{equation}\label{CMsMooredet1}
			\det(M(A_1,A_2,\ldots,A_{r};r))=\prod_{\mathbf{c}}(c_1\cdot A_1+\cdots+c_r\cdot A_r),
		\end{equation}
		where \(\mathbf{c}=(c_1, c_2, \ldots, c_r)\) runs over a complete set of direction vectors in \(\mathbb{F}_q^r\), and is specified by having the last non-zero entry equal to \(1\).
	\end{lemma}
	\begin{proof}
		Equation (\ref{CMsMooredet1}) holds if and only if the following equation is satisfied:
		\begin{equation}\label{CMsMooredet2}
			\det(M(a_1(x),a_2(x),\ldots,a_{r}(x);r))=\prod_{\mathbf{c}}(c_1\cdot a_1(x)+\cdots+c_r\cdot a_r(x)),
		\end{equation}
		where \(a_1(x), a_2(x), \ldots, a_r(x)\in\frac{\mathbb{F}_q[x]}{\langle x^N - 1\rangle}\) correspond to \(A_1, A_2, \ldots, A_r\) respectively, and the traversal range of \(\mathbf{c}\) remains unchanged.
		
		The remaining part of the proof employs the principle of mathematical induction on \(r\). The inspiration for this approach is drawn from \cite[Lemma 3.51]{Lidl1}. Based on the definition of the traversal range of \(\mathbf{c}\), Equation (\ref{CMsMooredet2}) is equivalent to
		\begin{equation}\label{CMsMooredet3}
			\det(M(A_1,A_2,\ldots,A_{r};r))=a_1(x)\prod_{i = 1}^{r - 1}\prod_{(c_1,\ldots,c_i)\in\mathbb{F}_q^i}\left(a_{i + 1}(x)-\sum_{i_1 = 1}^i c_{i_1}\cdot a_{i_1}(x)\right).
		\end{equation}
		For the sake of notational convenience, let \(R_1=\frac{\mathbb{F}_q[x]}{\langle x^N - 1\rangle}\), and denote \(\det(M(A_1,\ldots,A_{r};r))\) as \(M_r\). It is straightforward to verify that Equation (\ref{CMsMooredet3}) holds for \(r = 1\).
		Assume that Equation (\ref{CMsMooredet3}) holds for some \(r\geq1\). Consider the polynomial
		\begin{equation*}
			M(y)=\begin{vmatrix}
				a_1(x) & a_2(x) & \cdots & a_{n}(x) & y \\
				a_1^q(x) & a_2^q(x) & \cdots & a_{n}^q(x) & y^q \\
				\vdots& \vdots & \ddots & \vdots & \vdots \\
				a_1^{q^{r - 1}}(x) & a_2^{q^{r - 1}}(x) & \cdots & a_{n}^{q^{r - 1}}(x) & y^{q^{r - 1}} \\
				a_1^{q^{r}}(x) & a_2^{q^{r}}(x) & \cdots & a_{n}^{q^{r}}(x) & y^{q^{r}}
			\end{vmatrix}
		\end{equation*}
		defined over the ring \(R_1\).
		Expanding the determinant \(M(y)\) along its last row, we obtain
		\begin{equation}\label{CMsMooredet4}
			M(y)=M_ry^{q^r}+\sum_{i = 0}^{r - 1}\alpha_i y^{q^i}
		\end{equation}
		where \(\alpha_i\in R_1\). Since the characteristic of the finite field \(\mathbb{F}_q\) is identical to that of \(R_1\), we have
		\begin{equation*}
			\left(\sum_{i_1 = 1}^r c_{i_1}\cdot a_{i_1}(x)\right)^{q^r}=\sum_{i_1 = 1}^r c_{i_1}\cdot a_{i_1}^{q^r}(x).
		\end{equation*}
		Consequently, all linear combinations \(\sum_{i_1 = 1}^r c_{i_1}\cdot a_{i_1}(x)\) with \(c_{i_1}\in\mathbb{F}_q\) for \(i_1\in[r]\) are roots of the polynomial \(M(y)\) over \(R_1\).
		
		If \(a_1(x), a_{2}(x),\ldots,a_{r}(x)\) are linearly independent over \(R_1\), then \(M(y)\) has \(q^r\) distinct roots, and Equation (\ref{CMsMooredet4}) can be factored as
		\begin{equation}\label{CMsMooredet5}
			M(y)=M_r\prod_{(c_1,\ldots,c_r)\in\mathbb{F}_q^r}\left(y-\sum_{i_1 = 1}^r c_{i_1}\cdot a_{i_1}(x)\right).
		\end{equation}
		If \(a_1(x), a_{2}(x),\ldots,a_{r}(x)\) are linearly dependent over \(R_1\), then \(M_r = 0\). Therefore, $$M\cdot\sum_{i_1 = 1}^r c_{i_1}\cdot a_{i_1}(x) = 0,$$ and Equation (\ref{CMsMooredet5}) holds in all cases. Then, we have
		\begin{equation*}
			M_{r + 1}=M(a_{r + 1}(x))=M_{r}\prod_{(c_1,\ldots,c_r)\in\mathbb{F}_q^r}\left(a_{r + 1}(x)-\sum_{i_1 = 1}^r c_{i_1}\cdot a_{i_1}(x)\right).
		\end{equation*}
		By the principle of mathematical induction on \(r\), we arrive at the desired result.
	\end{proof}
	With the aid of Lemma \ref{CMsMoo}, we obtain the following result.
	\begin{theorem}\label{MimiLv1}
		Let $r$ and $m$ be positive integers with $r<m$.
		Let $a_1(x),a_2(x),\ldots,a_{m}(x)\in \frac{\F_q[x]}{\langle x^N-1\rangle}$. Let $A_i$ be the associated circulant matrices of $a_i(x)$ which has row weight of $t$ with some even integer $t$ and $i\in[m]$.  Assume that
		\begin{equation}\label{condition4}
			\gcd\left(a_{i_1}(x)\cdot\prod_{j=1}^{r-1}\prod_{(c_1,\ldots,c_{j})\in\F_q^j}\left(a_{i_{j+1}}(x)-\sum_{k_1=1}^jc_{k_1}a_{i_{k_1}}(x)\right),x^N-1\right)=x-1,
		\end{equation}
		for any $\{i_1,\ldots,i_{r}\}\in {[m]\choose r}$. Then the code with parity-check matrix
		\begin{equation}\label{Punc15}
			H=\left(
			\begin{array}{cccc}
				{\rm Pu}(A_1,1) & {\rm Pu}(A_2,1) & \cdots & {\rm Pu}(A_m,1)  \\
				{\rm Pu}(A_1^2,1) & {\rm Pu}(A_2^2,1) & \cdots & {\rm Pu}(A_m^2,1)  \\
				{\rm Pu}(A_1^4,1) & {\rm Pu}(A_2^4,1) & \cdots & {\rm Pu}(A_m^4,1)  \\
				\vdots & \vdots & \ddots & \vdots \\
				{\rm Pu}(A_1^{2^{r-1}},1) & {\rm Pu}(A_2^{2^{r-1}},1) & \cdots & {\rm Pu}(A_m^{2^{r-1}},1)  \\
			\end{array}
			\right)
		\end{equation}
		is an MDS array code with sub-packetization level $N-1$.
	\end{theorem}
	\begin{proof}
		Recall that in this subsection, we assume that the characteristic of $\mathbb{F}_q$ is $2$.
		For an arbitrary subset \(\{i_1,i_2,\ldots,i_r\}\in {[m]\choose r}\), we define the following square matrix:
		\begin{equation*}
			H_{\{A_{i_1},A_{i_2},\ldots,A_{i_r}\}}=\left(
			\begin{array}{cccc}
				A_{i_1} & A_{i_2} & \cdots & A_{i_r}  \\
				A_{i_1}^2 & A_{i_2}^2 & \cdots & A_{i_r}^2  \\
				A_{i_1}^4 & A_{i_2}^4 & \cdots & A_{i_r}^4  \\
				\vdots & \vdots & \ddots & \vdots \\
				A_{i_1}^{2^{r - 1}} & A_{i_2}^{2^{r - 1}} & \cdots & A_{i_r}^{2^{r - 1}}  \\
			\end{array}
			\right).
		\end{equation*}
		It is easy to check that among the vectors \(\mathbf{c}=(c_1(x),\ldots,c_{r}(x))\) in the module \(\left(\frac{\mathbb{F}_q[x]}{\langle x^N - 1\rangle}\right)^r\), the vectors \(\mathbf{c}\) that satisfy the equation \(\mathbf{c}\cdot H_{\{A_{i_1},A_{i_2},\ldots,A_{i_r}\}}=\mathbf{0}\) are limited to either the all-zero vector or the vector \(\mathbf{c}=\left(\underbrace{\frac{x^N-1}{x-1},\frac{x^N-1}{x-1},\ldots,\frac{x^N-1}{x-1}}_r\right)\).
		
		We then assert that for any \(\{i_1,i_2,\ldots,i_r\}\in {[m]\choose r}\), the matrix
		\begin{equation*}
			\overline{H_{\{A_{i_1},A_{i_2},\ldots,A_{i_r}\}}}=\left(
			\begin{array}{cccc}
				\overline{A_{i_1}} & \overline{A_{i_2}} & \cdots & \overline{A_{i_r}}  \\
				\overline{A_{i_1}^2} & \overline{A_{i_2}^2} & \cdots & \overline{A_{i_r}^2}  \\
				\overline{A_{i_1}^4} & \overline{A_{i_2}^4} & \cdots & \overline{A_{i_r}^4}  \\
				\vdots & \vdots & \ddots & \vdots \\
				\overline{A_{i_1}^{2^{r - 1}}} & \overline{A_{i_2}^{2^{r - 1}}} & \cdots & \overline{A_{i_r}^{2^{r - 1}}}  \\
			\end{array}
			\right)
		\end{equation*}
		has a rank of \(r(N - 1)\), where \(\overline{A_{i_j}^{2^{k-1}}}\) represents the truncated matrix of \(A_{i_j}^{2^{k}}\) with the last row removed with $k\in [r]$.
		
		Suppose, for the sake of contradiction, that there exists a non-zero vector \(\mathbf{c}^{\prime}=(c_1^{\prime},\ldots,c_{r}^{\prime})\in \left(\mathbb{F}_q^{N - 1}\right)^r\) such that
		$$\mathbf{c}^{\prime}\cdot \overline{H_{\{A_{i_1},A_{i_2},\ldots,A_{i_r}\}}}=\mathbf{0}\in\left(\mathbb{F}_q^{N - 1}\right)^r.$$
		We define \(\mathbf{c}^{\prime\prime}=((c_1^{\prime},0),\ldots,(c_{r}^{\prime},0))\in \left(\mathbb{F}_q^{N}\right)^r\). Then, we have
		$$\mathbf{c}^{\prime}\cdot \overline{H_{\{A_{i_1},A_{i_2},\ldots,A_{i_r}\}}}=\mathbf{0}\in\left(\mathbb{F}_q^{N - 1}\right)^r\longrightarrow\mathbf{c}^{\prime\prime}\cdot H_{\{A_{i_1},A_{i_2},\ldots,A_{i_r}\}}=\mathbf{0}\in\left(\mathbb{F}_q^{N}\right)^r.$$
		Based on the one-to-one correspondence between \(\mathbb{F}_q^{N}\) and \(\frac{\mathbb{F}_q[x]}{\langle x^N - 1\rangle}\), we conclude that \(\mathbf{c}^{\prime}\) can only be the all-zero vector, which contradicts our initial assumption.
		
		Since the row weight \(t\) of \(A_{i_j}^{2^{k}}\) is even, the last column of \(A_{i_j}^{2^{k}}\) can be linearly expressed in terms of the first \(N - 1\) columns. This indicates that the rank of the square matrix
		$$\left(
		\begin{array}{cccc}
			{\rm Pu}(A_{i_1},1) & {\rm Pu}(A_{i_2},1) & \cdots & {\rm Pu}(A_{i_r},1)  \\
			{\rm Pu}(A_{i_1}^2,1) & {\rm Pu}(A_{i_2}^2,1) & \cdots & {\rm Pu}(A_{i_r}^2,1)  \\
			{\rm Pu}(A_{i_1}^4,1) & {\rm Pu}(A_{i_2}^4,1) & \cdots & {\rm Pu}(A_{i_r}^4,1)  \\
			\vdots & \vdots & \ddots & \vdots \\
			{\rm Pu}(A_{i_1}^{2^{r - 1}},1) & {\rm Pu}(A_{i_2}^{2^{r - 1}},1) & \cdots & {\rm Pu}(A_{i_r}^{2^{r - 1}},1)  \\
		\end{array}
		\right)$$
		is \(r\cdot (N - 1)\) for any \(\{i_1,i_2,\ldots,i_r\}\in {[m]\choose r}\).
	\end{proof}
	Combining Lemma \ref{suff1} and Theorem \ref{MimiLv1}, we can directly obtain the following result.
	\begin{theorem}\label{thmsec}
		Let \(\C\) be a code that satisfies the conditions of Theorem \ref{MimiLv1}, and the notations are consistent with those in Theorem \ref{MimiLv1}. If for any \(\{i_1, i_2\} \in {[m] \choose 2}\) and \(\{i_3, i_4\} \in {[r] \choose 2}\), the sub-matrix $\left(
		\begin{array}{cc}
			A_{i_1}^{2^{i_3-1}} & A_{i_1}^{2^{i_4-1}}\\
			A_{i_2}^{2^{i_3-1}} & A_{i_2}^{2^{i_4-1}}\\
		\end{array}
		\right)$ satisfies all three conditions in Lemma \ref{suff1}, then \(\C\) is a block MDS LDPC code with girth $>4$.
	\end{theorem}
	If CM\((3)\) appears in the parity-check matrix \(H\), then its girth will definitely not exceed $6$ \cite[Theorem 18]{Smar2}. If CM with column weight greater than $3$ appears in the parity-check matrix \(H\), the occurrence of short cycles becomes even more difficult to avoid. Therefore, our main interest lies in the case of CM\((2)\).
	
	We consider the case of high code rate where \(q = 2\) and \(r = 2\). Xiao {\it et al.} \cite{Xiao1} also similarly considered the case with only two rows of blocks. We will subsequently compare our work with that of \cite{Xiao1}. Let \(N\) be an odd prime number such that \(2\) is a primitive root modulo \(N\) and \(N>3\). Let \(a_1(x),a_2(x),\ldots,a_m(x)\in\frac{\mathbb{F}_2[x]}{\langle x^N - 1\rangle}\) be binomials. The restriction (18) appearing in Theorem \ref{MimiLv1} will be simplified to
	$$\gcd\left(a_{i_1}(x)a_{i_2}(x)(a_{i_2}(x)-a_{i_1}(x)),x^N - 1\right)=x - 1$$
	for any \(\{i_1,i_2\}\in {[m] \choose 2}.\)
	Since \(2\) is a primitive root modulo \(N\), the polynomial \(x^{N - 1}+\cdots +x + 1\) is irreducible over \(\mathbb{F}_2\). Then, in the step of constructing the MDS array code, there is no restriction on the selection of \(a_1(x),a_2(x),\ldots,a_m(x)\) as long as they are distinct from each other.
	
	Next, we need to make \(\text{Index}(a_i(x))\) satisfy the three restrictive conditions described in Lemma \ref{suff1}. Since \(N\) is an odd prime integer, the first restrictive condition is obviously satisfied. For the second restrictive condition, since \(3\) is not a factor of \(N\), the restrictive condition in the vertical direction is obviously satisfied. The restrictive condition in the horizontal direction is transformed into \(\Delta (\text{Index}(a_{i_1}(x)))\cap \Delta (\text{Index}(a_{i_2}(x)))=\varnothing\) for any \(\{i_1,i_2\}\in {[m]\choose 2}\). The third restrictive condition can be simplified to $$0\notin(\text{Index}(a_{i_1}(x)) + 2\cdot\text{Index}(a_{i_2}(x)))- (\text{Index}(a_{i_2}(x)) + 2\cdot\text{Index}(a_{i_1}(x)))$$
	for any \(\{i_1,i_2\}\in {[m] \choose 2}.\)
	\begin{example}\label{ex17}
		Assume that $q=2$, $r=2$, and $N=239$. We select the following polynomials:
\begin{align*}
	a_1(x) & = x^{230} + x^{29}, &
	a_2(x) & = x^{167} + x^{36}, &
	a_3(x) & = x^{179} + x^{80}, \\
	a_4(x) & = x^{61} + x^{3}, &
	a_5(x) & = x^{121} + x^{25}, &
	a_6(x) & = x^{128} + x^{66}, \\
	a_7(x) & = x^{158} + x^{85}, &
	a_8(x) & = x^{143} + x^{38}, &
	a_9(x) & = x^{175} + x^{46}, \\
	a_{10}(x) & = x^{101} + x^{92}, &
	a_{11}(x) & = x^{159} + x^{95}, &
	a_{12}(x) & = x^{127} + x^{53}.
\end{align*}

It is easy to verify that the above polynomials satisfy the constraints in Theorem \ref{thmsec}. Therefore, we can construct the MDS array code $\C_{10}$ with sub-packetization level $N-1 = 238$. Regarding $\C_{10}$ as a linear code over $\mathbb{F}_2$, it can be used as a binary LDPC code with parameters $(2856, 2380)$, rate of $0.8333$ and girth $6$.

Xiao {\it et al.}  provided the simulation data of $\C_{2,12,239}^{RS}$ and $\C_{2,12,239}^{Gabidulin}$ in \cite{Xiao1}. Both codes have parameters $(2868, 2391)$ and rate of $0.8337$. From the simulation results, the LDPC code constructed by us with the same code rate and similar parameters has better performance. The girth of the code $\C_{2,12,239}^{RS}$ is $8$.

Chen {\it et al.} \cite{Chen11} proved that the girth is $12$ if and only if the parity-check matrix mentioned in \cite{Xiao1} has the modular Golomb ruler property.
In fact, to ensure the girth is $12$, for the LDPC code $\C_{2,12,239}^{Gabidulin}$ obtained by the Gabidulin-type, the sub-packetization level should be at least $239$ so that the found Markers can make the corresponding parity-check matrix have the modular Golomb ruler property. In other words, the girth of the code $\C_{2,12,239}^{Gabidulin}$ is $12$.

It is obvious that $N=239$ is not the minimum sub-packetization level for keeping the cardinality of $N^{\prime}$-modular Golomb ruler markers as $12$. We also take into account the case where the cardinality of the optimal $N^{\prime}$-modular Golomb ruler markers is $12$. As can be seen from \cite[Table I]{Xiao1}, if we want to control the number of block rows as $12$, the minimum sub-packetization level for each CM is \(N_{\text{optimal}} = 133\). At this time, $$GR(12, 133) = \{0, 1, 3, 12, 20, 38, 34, 81, 94, 88, 104, 109\}$$ is an optimal $133$-modular Golomb ruler. We use the construction method in the literature \cite{Xiao1} and take $GR(12, 133)$ as the selection set of CPMs, and the obtained LDPC code is denoted as $\C_{2,12,133}^{Xiao133}$, with parameters $(1596, 1331)$ and a code rate of $0.8334$. At this time, the girth of $\C_{2,12,133}^{Xiao133}$ is $12$.

We select the following polynomials:
\begin{align*}
	a_1(x) & = x^{33} + x^{35}, &
	a_2(x) & = x^{89} + x^{111}, &
	a_3(x) & = x^{60} + x^{86}, \\
	a_4(x) & = x^{46} + x^{116}, &
	a_5(x) & = x^{9} + x^{14}, &
	a_6(x) & = x^{10} + x^{39}, \\
	a_7(x) & = x^{13} + x^{115}, &
	a_8(x) & = x^{27} + x^{124}, &
	a_9(x) & = x^{52} + x^{59}, \\
	a_{10}(x) & = x^{30} + x^{131}, &
	a_{11}(x) & = x^{38} + x^{54}, &
	a_{12}(x) & = x^{98} + x^{104}.
\end{align*}
It is easy to verify that the above polynomials satisfy the constraints in Theorem \ref{thmsec}. Therefore, we can construct an MDS array code $\C_{11}$ with a sub-packetization level of \(N_{\text{optimal}} - 1 = 133 - 1 = 132\). This code can be regarded as an LDPC code with parameters $(1584, 1320)$ and a code rate of $0.833$. The girth of the constructed LDPC code is $6$, which is verified by Python. It can be seen from Figure \ref{Fig4} that the performance of the codes we constructed are better, even though the girth of these codes is only $6$.
		\begin{figure}[H]
			\centering
			\includegraphics[width=1.0\textwidth]{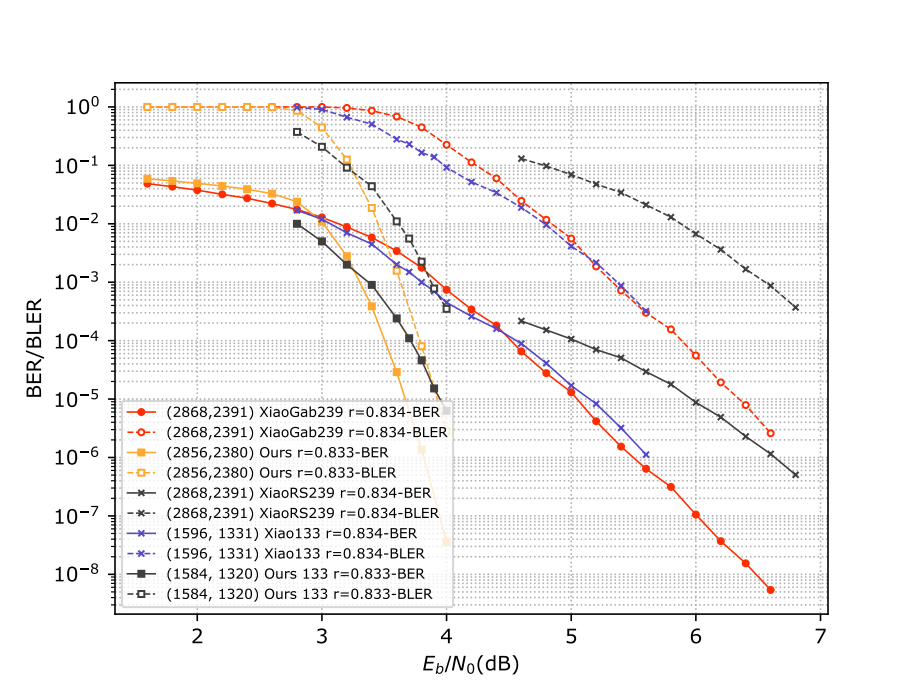}
			\caption{The BER and BLER performances of five LDPC codes}
			\label{Fig4}
		\end{figure}
	\end{example}
	\begin{example}\label{ex18}
		Assume that \( q = 2 \), \( r = 2 \), and \( N = 1453 \). We select the following polynomials:
\begin{align*}
a_1(x) & = x^{1071} + x^{726}, &
a_2(x) & = x^{1253} + x^{761}, &
a_3(x) & = x^{1322} + x^{63}, \\
a_4(x) & = x^{1322} + x^{63}, &
a_5(x) & = x^{1313} + x^{461}, &
a_6(x) & = x^{847} + x^{468}, \\
a_7(x) & = x^{886} + x^{568}, &
a_8(x) & = x^{630} + x^{247}, &
a_9(x) & = x^{1121} + x^{1060}, \\
a_{10}(x) & = x^{1241} + x^{56}, &
a_{11}(x) & = x^{814} + x^{797}, &
a_{12}(x) & = x^{557} + x^{553}, \\
a_{13}(x) & = x^{1397} + x^{1239}, &
a_{14}(x) & = x^{871} + x^{812}, &
a_{15}(x) & = x^{1151} + x^{810}, \\
a_{16}(x) & = x^{1138} + x^{1074}, &
a_{17}(x) & = x^{1319} + x^{383}, &
a_{18}(x) & = x^{832} + x^{677}, \\
a_{19}(x) & = x^{927} + x^{198}, &
a_{20}(x) & = x^{449} + x^{42}, &
a_{21}(x) & = x^{1424} + x^{1049}, \\
a_{22}(x) & = x^{855} + x^{86}, &
a_{23}(x) & = x^{437} + x^{375}, &
a_{24}(x) & = x^{1326} + x^{106}, \\
a_{25}(x) & = x^{1158} + x^{983}, &
a_{26}(x) & = x^{1430} + x^{1380}, &
a_{27}(x) & = x^{639} + x^{166}, \\
a_{27}(x) & = x^{625} + x^{43}, &
a_{28}(x) & = x^{1358} + x^{428}, &
a_{29}(x) & = x^{518} + x^{212}, \\
a_{30}(x) & = x^{876} + x^{762}, &
a_{31}(x) & = x^{791} + x^{349}.
\end{align*}
It is easy to verify that the above polynomials satisfy the constraints in Theorem \ref{thmsec}. Therefore, we can construct the MDS array code \( \C_{12} \) with the sub-packetization level \( N - 1 = 1452 \). Regarding \( \C_{12} \) as a linear code over \( \mathbb{F}_2 \), it can be used as a binary LDPC code with parameters \( (45012, 42108) \) and a code rate of \( 0.935 \).

We select the LDPC code \( B_{\text{PaG},Q}(4, 61) \) given in Example 1 of \cite{Li1}, which has parameters \( (44713, 41784) \) and a code rate of \( 0.935 \). As shown in Figure 5, the LDPC code constructed in this paper outperforms \( B_{\text{PaG},Q}(4, 61) \) when the SNR exceeds \( 4.6 \, \text{dB} \).
		\begin{figure}[H]
			\centering
			\includegraphics[width=0.74\textwidth]{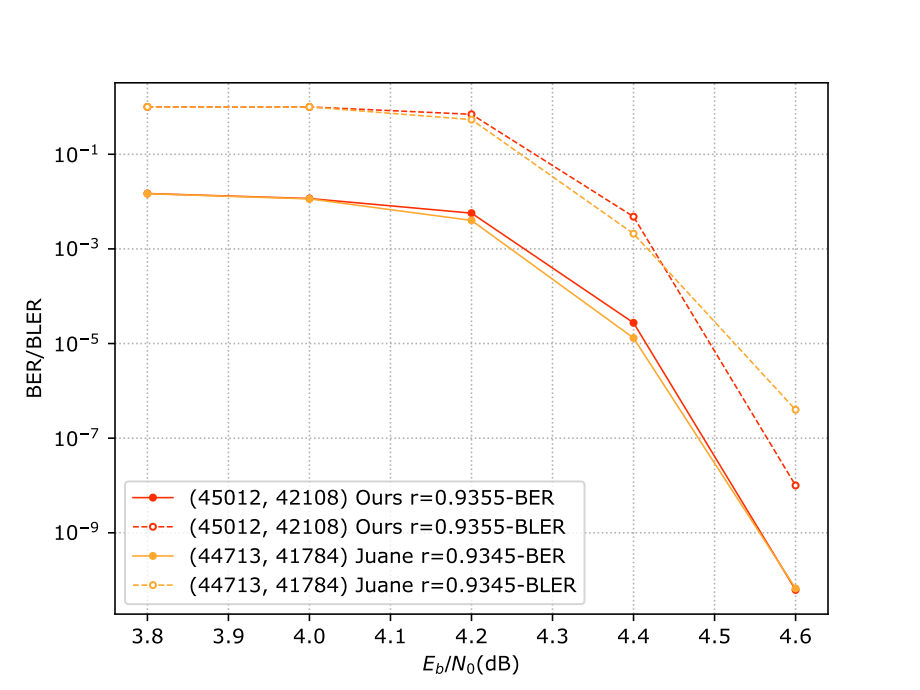}
			\caption{The BER and BLER performances of $\C_{11}$ and $B_{\text{PaG},Q}(4, 61)$.}
			\label{The BER and BLER performances of}
		\end{figure}
	\end{example}
	
	\begin{example}\label{ex19}
Assume that \( q = 2 \), \( r = 2 \), and \( N = 163 \). We select the following polynomials:
\begin{align*}
a_1(x) & = x^{100} + x^{126}, &
a_2(x) & = x^{70} + x^{128}, &
a_3(x) & = x^{61} + x^{97}, \\
a_4(x) & = x^{49} + x^{50}, &
a_5(x) & = x^{34} + x^{137}, &
a_6(x) & = x^{60} + x^{132}, \\
a_7(x) & = x^{101} + x^{143}, &
a_8(x) & = x^{42} + x^{103}, &
a_9(x) & = x^{37} + x^{76}, \\
a_{10}(x) & = x^{134} + x^{159}, &
a_{11}(x) & = x^{85} + x^{88}, &
a_{12}(x) & = x^{24} + x^{154}.
\end{align*}
It is easy to verify that the above polynomials satisfy the constraints in Theorem \ref{thmsec}. Therefore, we can construct the MDS array code \( \C_{12} \) with the sub-packetization level \( N - 1 = 162 \). Regarding \( \C_{12} \) as a linear code over \( \mathbb{F}_2 \), it can be used as a binary LDPC code with parameters \( (1944, 1620) \) and a code rate of \( 0.83 \).

We select the QC-LDPC code 11n-D2-1944b-R56 from IEEE 802.11n \cite{IEEE} for comparison. The code 11n-D2-1944b-R56 has parameters \( (1944, 1620) \) and a code rate of \( 0.83 \). As shown in Figure \ref{Fig7}, the LDPC code constructed in this paper outperforms 11n-D2-1944b-R56 when the SNR exceeds \( 4.0 \, \text{dB} \).
		\begin{figure}[H]
			\centering
			\includegraphics[width=0.84\textwidth]{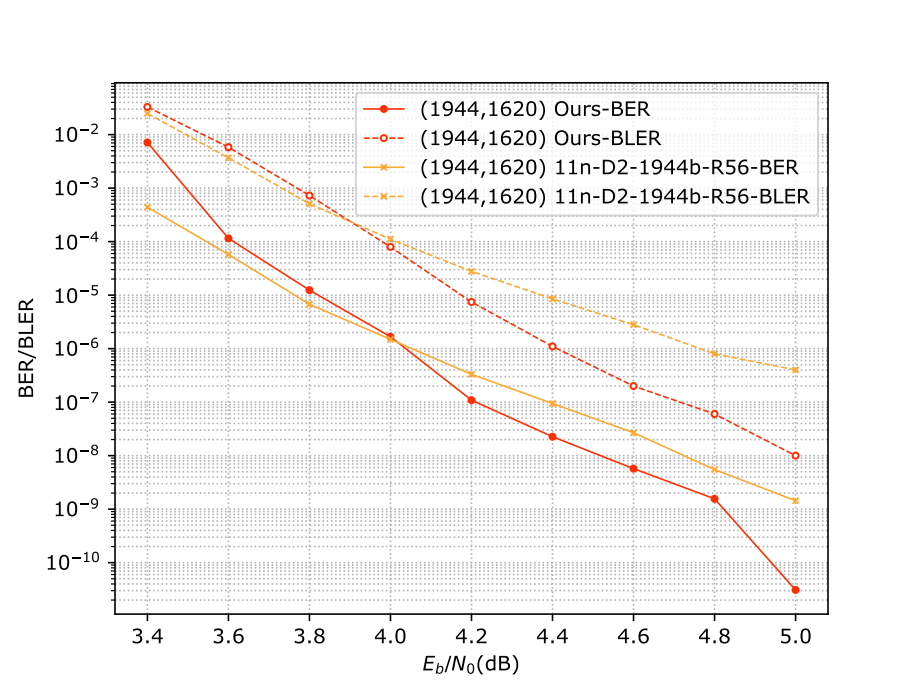}
			\caption{The BER and BLER performances of $\C_{12}$ and 11n-D2-1944b-R56}
			\label{Fig7}
		\end{figure}
	\end{example}
	
	\section{Conclusion and future work}
In this paper, explicit constructions for three classes of block MDS LDPC codes are presented based on several types of punctured circulant matrices. The three constructions exhibit the following characteristics:
\begin{itemize}
\item [{\rm (1)}] The first construction is based on the PUCPM. It features a wide range of applicability without restrictions on the finite field, allowing the code length to reach the same level as the sub-packetization
level. The overall performance of the code can be further enhanced by algebraically selecting matrices to improve the girth. Although the algebraic structure of punctured quasi-cyclic codes is partially damaged, efficient encoding and decoding algorithms have been proposed in \cite{Lv1,Fang1}, somewhat diminishing this advantage.

\item [{\rm (2)}] The second construction relies on the CSM. Its applicability is limited (only applicable to certain non-binary fields), yet it can also achieve a code length equivalent to the sub-packetization
level. The advantage lies in retaining the full algebraic structure of the quasi-cyclic code without the need for puncturing.

\item [{\rm (3)}] The third construction is based on the PUCM$(t)$, where $t>1$ and PUCM$(t)$ requires a Moore-type structure. LDPC codes constructed under this construction outperform many current LDPC codes, including those in the IEEE standards.
\end{itemize}
All three classes of constructions can be applied to optimize quantum key distribution systems. Block MDS LDPC codes are particularly suited for channels with mixed error types (random errors + burst errors), where appropriate decoding schemes can be employed based on the error type.\\

The future works of this topic will focus on the following points:
\begin{itemize}
  \item [{\rm (1)}] Bidirectional Moore-type matrices are applied to construct low-rate (rate $<0.5$) LDPC codes and convolutional codes \cite{Tanner2}. These codes also possesses excellent algebraic structures and and easily determinable girth.
 By applying similar puncturing operations and controlling the corresponding greatest common divisor constraints, it is expected to construct new block MDS LDPC codes and further explore the related constructions of MDS convolutional codes.
  \item [{\rm (2)}]LDPC codes based on Vandermonde-type and Moore-type matrices in this paper do not contain all-zero blocks. Masking operations are commonly used to optimize LDPC codes. However, blindly ``puncturing'' block MDS LDPC codes will most likely result in the loss of the block MDS property in the newly generated LDPC codes. In the future, we may attempt to develop an improved masking scheme that preserves the block MDS property by combining guidance from generalized RS array codes with density evolution algorithms.
\end{itemize}
	
	\section*{Acknowledgement}
	This research is supported by the China Postdoctoral Science Foundation(2024M751606), the National Natural Science Foundation of China (62171248,62401144), the Pengcheng National Laboratory Key Project (PCL2021A07), the Key Area Research and Development Program of Guangdong Province (2020B0101110003), and the Guangdong Basic and Applied Basic Research Foundation under grant (2021A1515110066).

\end{document}